\documentclass{article}
\usepackage{ijcai19}

\usepackage{times}
\usepackage{soul}
\usepackage{url}
\usepackage[hidelinks]{hyperref}
\usepackage[utf8]{inputenc}
\usepackage[small]{caption}
\usepackage{graphicx}
\usepackage{amsmath}
\usepackage{booktabs}
\usepackage{algorithm}
\usepackage{algorithmic}
\usepackage{latexsym}

\PassOptionsToPackage{table}{xcolor}
\usepackage{graphics}
\usepackage{url}
\usepackage{enumerate}

\usepackage{amsmath}
\usepackage{amsthm}
\usepackage{amssymb}
\usepackage{thmtools,thm-restate}
\usepackage{graphicx}
\usepackage{tikz}
\usetikzlibrary{arrows,backgrounds,positioning}
\usepackage{colortbl}

 \newtheorem{corollary}{Corollary}
 \newtheorem{lemma}{Lemma}
 \newtheorem{example}{Example}
 \newtheorem{claim}{Claim}

\newcommand{\name}[1]{graphical matching problem}

\newcommand{\obar}[1]{\overline{#1}}

\newcommand{\alloc}{M}
\newcommand{\pgraph}{G}
\newcommand{\igraph}{H}
\newcommand{\sw}{\mathit{SW}}
\newcommand{\txtblue}{\mathit{blue}}
\newcommand{\txtred}{\mathit{red}}
\newcommand{\cycle}{cycle}

\title{Graphical One-Sided Markets}
\author{
Sagar Massand
\And
Sunil Simon
\affiliations
Department of CSE, IIT Kanpur, Kanpur, India\\
\emails
\{smassand, simon\}@cse.iitk.ac.in,
}

\begin{document}

\maketitle

\begin{abstract}

We study the problem of allocating indivisible objects to a set of
rational agents where each agent's final utility depends on the
intrinsic valuation of the allocated item as well as the allocation
within the agent's local neighbourhood. We specify agents' local
neighbourhood in terms of a weighted graph. This extends the model of
one-sided markets to incorporate neighbourhood externalities. We
consider the solution concept of stability and show that, unlike in the
case of one-sided markets, stable allocations may not always
exist. When the underlying local neighbourhood graph is symmetric, a
2-stable allocation is guaranteed to exist and any decentralised
mechanism where pairs of rational players agree to exchange objects
terminates in such an allocation. We show that computing a 2-stable
allocation is PLS-complete and further identify subclasses which are
tractable. In the case of asymmetric neighbourhood structures, we show
that it is NP-complete to check if a 2-stable allocation exists. We then
identify structural restrictions where stable allocations always exist
and can be computed efficiently.
Finally, we study the notion of envy-freeness in this framework.

%% which ensures the existence of a
%% stable outcome which can be efficiently computed.

%% We show that when the underlying local neighbourhood graph is
%% symmetric, a 2-stable allocation always exists and any decentralised
%% mechanism where players agree to exchange object terminate in a
%% 2-stable allocation. Although computing such a 2-stable allocation in
%% general is PLS-complete, we identify restrictions where efficient
%% computation is possible. When the local neighbourhood structure is not
%% symmetric, 2-stable outcomes need not exists. We show that it is
%% NP-complete to check if a 2-stable outcome exists in general. We
%% identify structural restrictions on the neighbourhood structure which
%% ensures the existence of a 2-stable (and even a core-stable) outcome.

\end{abstract}

\maketitle

\section{Introduction}
\label{sec:intro}
Allocation of indivisible items to rational agents is a central
problem which has been studied in economics and computer science. The
problem arises in a wide range of applications: including resource
allocation in distributed systems, spectrum allocation, kidney
exchange programs and so on. Stability and fairness are two
influential concepts which are studied in the context of allocation
problems. Various notions of fairness such as proportionality,
envy-freeness and maximin share guarantee \cite{Bud11} have been
studied in the general setting of resource
allocation \cite{BCM16,beynier2018local}. On the other hand, stability
as a solution concept has received greater attention in the context of
matchings \cite{GI89,roth1990random}.

In this paper, we introduce a model which is closely aligned to the
one-sided market, also known as the Shapley-Scarf housing
market \cite{SS74}. This is a fundamental and well studied framework
used to model an exchange economy. It consists of a set of rational
agents each owning a house along with a strict preference ordering
over all the houses present in the market. A stable allocation
corresponds to an assignment of houses to agents 
such that no coalition of agents can improve by internal redistribution.
%that conforms to their desirability criteria.
An important question is whether stable
allocations always exist in such markets and whether a finite sequence
of exchanges can converge to such an allocation. It was shown that
using a simple and efficient procedure often termed as Gale's Top
Trading Cycle, one can always find an allocation that is stable. The
allocation constructed in this manner also satisfies various desirable
properties like it being strategy-proof and Pareto optimal.

In one-sided markets, agents' preferences are typically assumed to be
strict and are not dependent on the allocation received by other
agents. However, in many practical situations, an agent's utility for
an allocation is dependent on both their intrinsic valuation for the
allocated item as well as the items which are allocated to the agents
within their neighbourhood.
We consider a framework to model the allocation problem where agents
have cardinal utilities associated with each allocation. The agents
are not allowed to make monetary payments to each other, thus the
utilities are non-transferable. For each agent $i$, the final utility
depends on two components: the intrinsic value associated with the
item assigned to agent $i$ and the items assigned to the agents within
the neighbourhood of agent $i$. We model the agents' neighbourhood
structure as a directed weighted graph in which the nodes correspond
to agents. For each agent, the incoming edges along with the
associated weights denote the quantitative ``influence'' that
neighbours have on the agent. In our model, this influence is also
item specific. That is, each item $a$ is associated with a subset of
compatible items $a_{\lambda}$. The externality for an agent who is
allocated an item $a$ depends on the influence from agents in its
neighbourhood who are assigned items from $a_{\lambda}$. 

Such a framework naturally captures various instances of allocation
problems that arise in practice. For instance, consider a housing
allocation problem, where agents have some intrinsic valuation
over the houses.
%  based on various parameters like quality, proximity to
% ***, living conditions and so on. 
In addition, suppose some of these agents are friends and would prefer
to reside in houses near each other. In such a scenario, the final
utility associated with an allocation depends on both the valuation of
the house as well as the owners of the neighbouring houses. Similarly,
consider an organization which needs to decide project allocation for
its incoming employees, with each project being mentored by senior
employees. An employee's preference for a project would rely on both,
the nature of the project and the mentors assigned for the project.

Certain structural aspects are explicated in the model and these
choices are influenced by the various frameworks often used to analyse
strategic interaction. For instance, graphical dependency structures
are known to play an important role in the analysis of strategic
interaction in terms of the existence of stable outcomes and their
computational properties. In game theory, such structures are studied
in graphical games \cite{KLS01}. While we do incorporate neighbourhood
externalities in our model, these are required to satisfy a pairwise
seperability constraint which specifies the influence of each agent in
the neighbourhood.  The eventual externality is additive over these
pairwise values. In game theory, a similar constraint on the utility
functions give rise to a well studied class of games called
\textit{polymatrix games} \cite{Jan68}. Imposing pairwise
separability is a natural restriction which typically helps achieve
better computational properties and also helps in showing stronger
lower bounds \cite{CD11,DFS14,RS15}.

%\noindent{\bf Related work.} 
Markets form a fundamental model of
exchange economy and allocations of indivisible items in such markets
have been studied extensively both in terms of stability and fairness
of allocation. 
In the context of two sided markets, the influence of neighbourhood
relations has been studied in \cite{AV09,Hoe13}.\cite{Hoe13} proposes
a framework where players are modelled as nodes in a social network
and explore possible matches only among the players in the current
neighbourhood. \cite{AV09} models a job market as a 2-sided matching
market and studies the importance of social contacts in terms of
stable allocation. \cite{ABH13} studies stable matchings in the
presence of social contexts where players are positioned in an
underlying weighted directed graph and the weights on the edges denote
the rewards for being matched to each agent. Each agent's final
utility for a matching depends on both the agent's individual reward
as well as the rewards received by the neighbours under the matching.

For one-sided markets, \cite{FD-graph17} introduces a framework that
models constraints based on the dependency relation between items with
additive utility functions.  The dependency is encoded as a graph on
the item set and an allocation is required to satisfy the constraint
that the items allocated to each agent forms a sub-graph of the item
graph. 
The paper looks at additive utility functions and investigates the
complexity of finding optimal allocations in terms of fairness,
envy-freeness and maximin share guarantee.
% \cite{Bud11}. 
\cite{LT18} further studies the computational properties of this model for maximin share allocations in the specific case when the item
graph forms a \cycle{}. Approximations of envy-freeness in terms of
EF-1 is considered in \cite{EF1-graph}. \cite{CEM17} studies
convergence properties for fair allocation of dynamics involving
exchange of items by rational agents in the presence of a social
network structure on the players. In this line of work listed above,
fairness is the main solution concept that is studied. In our model,
rather than viewing the social network structure as imposing
restrictions on the set of feasible allocations, we interpret the
neighbourhood structure as specifying agents' externalities which
eventually contributes towards the final utility. In such a setting,
the dynamics involving exchange of items are quite natural and
therefore, stability of an allocation is an important consideration
that we study.

Decentralised swap dynamics where pairs of agents exchange items or
services and the optimality of allocations is studied
in \cite{maudetSwapPower15}. \cite{gourves2017object} examines the
influence of a social network structure on players in terms of these
exchanges and optimal
allocations. \cite{todo15exchange,todocorecomplexity15} look at
unrestricted exchange dynamics in the context of lexicographic
preference ordering, with \cite{todo15exchange} focussing on player
preferences which have a common structure, implying an asymmetry in
the objects. As part of our work, we consider such swap dynamics in
the context of a neighbourhood structure which influences a player's
utility.

The presence of externalities has been studied in the
literature, mainly in the context of fair and efficient
allocations. \cite{BPZ13} studies the problem of fair division of
divisible heterogeneous resources in the presence of externalities;
where the agent's utility depends on the allocation to others. A
recent paper \cite{GSS18} considers a model similar to the one we
study in which they analyse fair allocation in terms of the maximin
share guarantee. \cite{lesca2018service} incorporates a restricted
notion of externality in the context of a service exchange model and
analyses the complexity of computing efficient allocation.

Another related line of work involves agent based modelling of
Schelling segregation. \cite{CLM18} proposed a game theoretic
framework to analyse Schelling segregation in which players are
partitioned into two types and are assigned nodes in an underlying
graph. Each agent's utility depends on the number of neighbours of
same type and the possibility of the agent having a favourite node in
the graph. \cite{EGI19} studies the existence of Nash equilibrium in
an extension of this model with arbitrary types of agents where
agents' unilateral deviation involve changing their location to an
unoccupied node in the underlying graph. The notion of swap stability
for Schelling segregation, which is closer to the notion of stability
that we study in this paper, is considered in \cite{AEG19}. While the
framework is related, the results are not directly applicable in the
setting of one-sided markets.

\section{The Model}
\label{sec:model}
We introduce the resource allocation problem that we study in this
paper which we call the \textit{graphical matching problem}.
%%  We study the resource allocation problem involving indivisible
%% goods where each agent is assigned exactly one item.
Let $N = \{1,2,\ldots,n\}$ be a finite set of agents (or players) and $A$ be a finite set
of items such that $|N| = |A|$. The local neighbourhood structure for
the set of agents, referred to as the \textit{player graph} is given
by a weighted directed graph $G = (N, \tau, w)$ where $\tau \subseteq
N \times N$ and $w$ is a function that associates with each edge
$(i,j) \in \tau$, a weight $w_{i,j} \in \mathbb{R}$. We say that the
player graph $G$ is \textit{symmetric} if for each pair of players
$i,j$, $(i,j) \in \tau$ implies $(j,i) \in \tau$ and
$w_{i,j}=w_{j,i}$. We say that the player graph $G$
is \textit{unweighted} if for all $(i,j) \in \tau$, $w_{i,j}=1$.

The dependency structure associated with the items, also referred to
as the \textit{item graph}, is specified by an undirected graph $H =
(A, \lambda)$ where $\lambda \subseteq A \times A$. We assume that
both $G$ and $H$ do not have self loops. For a player $i \in N$, let
$i_{\tau} = \{j\mid(j,i) \in \tau\}$ denote the set of all neighbours
of $i$ in $G$. For an item $a \in A$, let $a_{\lambda}
= \{b\mid(a,b) \in \lambda\}$ denote the set of all items connected to
$a$ according to the dependence structure $H$.

Each player has a valuation function $v_i: A \rightarrow
\mathbb{R}_{\geq 0}$ that specifies the initial utility that the
player associates with each item. An allocation $\pi: N \to A$ is a
bijection that assigns items to players. In other words, an allocation
$\pi$ assigns exactly one item to each player. Let $\Pi$ denote the
set of all allocations. For a player $i \in N$ and an allocation
$\pi \in \Pi$, let $N(i,\pi)=\{j \in
N \mid \pi(j) \in (\pi(i))_\lambda\}$. For $i,j \in N$, we say that $i$
is connected to $j$ in $\pi$ if $j \in N(i,\pi)$. Let
$d_i(\pi)=N(i,\pi) \cap i_\tau$.
%% For an allocation $\pi \in \Pi$, let $d_i(\pi)=\{j \in \tau_i \mid
%% \pi(j) \in \lambda_{\pi(i)}\}$.
The \textit{utility} of player $i$ for allocation $\pi$ is then given
by: $u_i(\pi)=v_i(\pi(i))+r_i(\pi)$ where $r_i(\pi)=\Sigma_{j \in
d_i(\pi)}w_{j,i}$. In other words, the utility of player $i$ on an
allocation $\pi$ depends on his valuation for $\pi(i)$ as well as the
agents that are assigned items in the local neighbourhood of $\pi(i)$ ($N(i, \pi)$).
To simplify notation, we often use $v_i(\pi)$ to denote
$v_i(\pi(i))$.  The social welfare for an allocation $\pi$ is defined
as $\sw(\pi)=\Sigma_{i \in N} u_i(\pi)$.  We say that the \name{} has
uniform valuation if for all $i \in N$ and $a \in A$, $v_i(a)=c$ for
some constant $c$. 
%% In a \name{} with uniform valuation, 
In this case, the utilities of players are determined solely by their
local neighbourhood structure. An instance of the \textit{graphical
matching problem} is specified by the tuple
$M=(\pgraph,\igraph,(v_i)_{i \in N})$ where $\pgraph=(N,\tau,w)$ and
$\igraph=(A,\lambda)$.  We say that 
%% an instance of the \name{},
$M=(\pgraph,\igraph,(v_i)_{i \in N})$ has symmetric neighbourhood if
$G$ is symmetric. $M$ is unweighted if $G$ is unweighted.

%\medskip

%\noindent {\bf Stability.} 
\textit{Stability} in an allocation captures the property that
players do not have any incentive to exchange goods and deviate to a
new allocation. Given an allocation $\pi$, we call a pair of players
$i,j \in N$ to be a blocking pair in $\pi$, if there exists another
allocation $\pi'$ where $\pi'(i) = \pi(j)$, $\pi'(j) = \pi(i)$ and
$\pi'(k) = \pi(k)$ for all $k \in N-\{i,j\}$ such that $u_i(\pi) <
u_i(\pi')$ and $u_j(\pi) < u_j(\pi')$. In other words, the players $i,j$ have an incentive to deviate from the current allocation by
exchanging their items. We say an allocation $\pi$ is
\textit{2-stable} if there is no blocking pair in $\pi$. 

Given an allocation $\pi$ along with a blocking pair $(i,j)$, we say
that $\pi'$ is a \textit{resolution} of the blocking pair if
$\pi'(i)=\pi(j)$ and $\pi'(j)=\pi(i)$. We denote this by
$\pi \to_{i,j} \pi'$. An
\textit{improvement path} is a maximal sequence of allocations
$\pi^0 \pi^1 \pi^2 \ldots$ such that for all $k \geq 1$,
$\pi^{k-1} \to_{i,j} \pi^k$ for some pair of players $i,j \in N$. It
is easy to observe that the existence of a finite improvement path
implies the existence of a 2-stable allocation.

For an allocation $\pi$, we call $X \subseteq N$ a blocking coalition,
if there exists $\pi'$ and a bijection $\mu:X \to X$ such that for all
$i \in X$, $\pi'(i) \neq \pi(i)$, $\pi'(i)=\pi(\mu(i))$ and $u_i(\pi')
> u_i(\pi)$. We say an allocation $\pi$ is \textit{core stable} if
there is no blocking coalition in $\pi$.

\begin{example}
\label{ex:ex1}
{\rm
Let the set of players $N=\{1,\ldots,6\}$ and $A=\{A,B,C,D,E,F\}$ with $v_i(a) = 0$ for all $i \in N$ and $a \in A$. Let
$\pgraph=(N,\tau,w)$ be defined as follows: $\tau =\{ (1,2), (2,3),
(3,1), (4,1), (5,2), (6,3)\}$ and $w_{i,j}=1$ for all
$(i,j) \in \tau$. Let the item graph be the structure given in
Figure~\ref{fig:no-2stable}. Consider the allocation $\pi$ where
$\pi(1)=A, \pi(2)=B, \pi(3)=C$ and $\pi(4)=D, \pi(5)=E, \pi(6)=F$,
then $u_1(\pi)=u_2(\pi)=u_3(\pi)=1$ and
$u_4(\pi)=u_5(\pi)=u_6(\pi)=0$. Now suppose the pair $(1,4)$ exchanges
items and let $\pi'$ be the resulting allocation. That is, $\pi'(4)=A$
and $\pi'(1)=D$. Then, $u_1(\pi')=0$.  }
\end{example}

%\noindent{\bf One-sided market.} 
The classical \textit{one-sided matching market} can be viewed as a special
case of the model given above.  Consider an instance of the one-sided
market consisting of two finite sets $X$ and $Y$ where $|X| =
|Y|$. Each $x \in X$ has a valuation function $z_x: Y \to \mathbb{R}$
which specifies the preference ordering over outcomes in $Y$. An
allocation $\pi$ is a bijection $\pi:X \to Y$. The notion of a
2-stable and core-stable allocation remains the same as defined
earlier; we view the valuation function $z_x$ as specifying the final
utility of player $x \in X$. Given an instance $M$ of a one-sided
market, we can construct an instance $M'$ of the \name{} such that the set
of stable allocations in $M$ precisely constitutes the set of stable
allocations in $M'$. We take $N=X$, $A=Y$ and
$v_i(\pi(i))=z_i(\pi(i))$ with the player graph $G=(N,\tau,w)$ where $\tau
= \emptyset$.

In one-sided markets, a 2-stable allocation always exists and it can
be computed using the Top Trading Cycle procedure. On the other hand,
the presence of neighbourhood externalities
%% , even with restricted pairwise separable constraints, 
results in dynamics that is more complex and significantly different
in terms of players' behavioural aspects. The example given below
shows that in the \name{}, a 2-stable outcome need not always exist
even in the simple case when the instance has uniform valuation and
when the underlying neighbourhood structure is unweighted.

\begin{figure}%[htbp]
\centering
\tikzstyle{agent}=[circle,draw=black!80,thick, minimum size=1.8em,scale=0.8]
\scalebox{0.92}{
\begin{tikzpicture}[auto,>=latex',on grid]
%shorten >=1pt,
%\begin{scope} \node[agent,label=above:{\small $\{\bot\}$}](x1){};
%\node[agent, right of=x1, node distance=1.5cm,label=above:{\small
%$\{\top, \bot\}$}](x2){T};
\newdimen\R
\R=1cm
\newcommand{\llab}[1]{{\small $\{#1\}$}}
%\begin{tikzpicture}[scale=4,cap=round,>=latex] 
% Radius of regular polygons
\draw (90: \R) node[agent] (1) {A};
\draw (90-120: \R) node[agent] (2) {B};
\draw (90-240: \R) node[agent] (3) {C};

\node[agent,right = 2.5 of 1] (4) {D};
\node[agent,right = 1 of 2] (6) {F};
\node[agent,right = 2.5 of 2] (5) {E};

    \draw[-] (1) -- (2) node [midway,left]{\small{}};
    \draw[-] (3) -- (1) node [midway,right]{\small{}};
    \draw[-] (2) -- (3) node [midway,above]{\small{}};

    \draw (4) -- (5) node [midway,left]{\small{}};
    \draw (6) -- (4) node [midway,right]{\small{}};
    \draw (5) -- (6) node [midway,above]{\small{}};

\end{tikzpicture}
}
\caption{\label{fig:no-2stable} An item graph}
%% \caption{A weighted \cchain{} with no Nash equilibrium. \label{fig:necklace-weight}}
\end{figure}
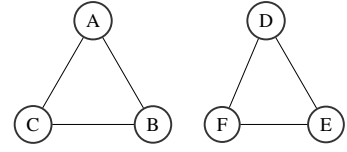

%\begin{restatable}{example}{exNoTwoStable}
\begin{example}
\label{ex:no-2stable}
{\rm Let the set of players $N=\{1,\ldots,6\}$ and
$A=\{A,B,C,D,E,F\}$. Suppose the player graph $\pgraph=(N,\tau,w)$
forms a \cycle{} on $N$ consisting of the edges $(6,1) \in \tau$
and $(i,i+1) \in \tau$ for all $i:1\leq i \leq 5$. For every edge
$(i,j) \in \tau$, let $w_{i,j}=c$ for some positive constant $c$. We
also assume uniform valuation $v_i(a)=0$ for all $i \in N$ and $a \in
A$. Let the item graph be the structure given in
Figure~\ref{fig:no-2stable}.  Since the item graph consists of two
3-cliques, it is sufficient to divide players into groups of 3, assign
them to any of the 3-cliques. It can be verified that this instance
does not have a 2-stable allocation. 
%(a complete enumeration can be found in the supplementary material).  
%%
%% Since the item graph consists of
%% two 3-cliques, it is sufficient to divide players into groups of 3,
%% assign them to any of the 3-cliques, and check if, for any group of 3,
%% we have a 2-stable allocation. Below we provide the list of all such
%% allocations and underline the blocking pair: $((\underline{1},2,3),
%% (\underline{4},5,6))$, $((\underline{1},2,4),(\underline{3},5,6))$,
%% $((\underline{1},2,5),(\underline{3},4,6))$,
%% $((1,2,\underline{6}),(\underline{3},4,5))$,
%% $((\underline{1},3,4),(2,\underline{5},6))$,
%% $((\underline{1},3,5),(2,\underline{4},6))$,
%% $((1,3,\underline{6}),(\underline{2},4,5))$,
%% $((1,\underline{4},5),(\underline{2},3,6))$,
%% $((1,\underline{4},6),(\underline{2},3,5))$,
%% $((1,\underline{5},6),(\underline{2},3,4))$.
%% %% \red{Provide some intuition on why 2-stable does not exists. Provide the enumeration of strategies in the Appendix.}
}
\end{example}
%\end{restatable}

\section{Stability in Symmetric Neighbourhood}
\label{sec:symmetric}
%% Example \ref{ex:no-2stable} shows that even for a \name{} with uniform
%% valuation and a simple neighbourhood structure consisting of a cycle,
%% a 2-stable allocation need not always exist. Thus a natural question
%% is to identify restricted classes of the \name{} where stable outcomes are
%% guaranteed to exist. In this section, we show that when the
%% underlying player graph is symmetric, a 2-stable outcome always
%% exists. We show that in general, it is PLS-hard to compute such a
%% stable allocation. We also identify restrictions under which stable
%% allocations can be computed efficiently.

Given Example \ref{ex:no-2stable}, a natural question is to identify
restricted classes of the \name{} where stable outcomes are guaranteed
to exist. We show that when the underlying player graph is symmetric,
a 2-stable outcome always exists. In general, it is PLS-hard to
compute such a stable allocation. We also identify restrictions under
which stable allocations can be computed efficiently.

\begin{restatable}{theorem}{thmSymPotential}
%% \begin{theorem}
\label{thm:sym-potential}
Every improvement path in a symmetric \name{} is finite. Thus, a
2-stable allocation always exists.
%There exists a 2-stable allocation for symmetric \name{}.
%\end{theorem}
\end{restatable}

\begin{proof}[Proof sketch]
We can argue that the following function acts as a potential
function. 
%in the symmetric \name{}: $\phi(\pi)
$\phi(\pi)= \sum_{i \in N} (v_i(\pi) + u_i(\pi))$.
\end{proof}

%% An immediate consequence of Theorem~\ref{thm:sym-potential} is that in
%% a symmetric \name{}, a 2-stable allocation always exists.

\begin{corollary}
In a symmetric \name{} with uniform valuation where the underlying player graph
is unweighted, we can compute a 2-stable allocation in polynomial
time.
\end{corollary}
\begin{proof}[Proof sketch]
Suppose $v_i(a)=c$ for all $i \in N$ and $a \in A$. Then, $\phi$ is
bounded above by $2nc + 2|\tau|$ and below by $2nc$; in each
resolution step the value increases by at least 2.
\end{proof}
%% \begin{proof}[Proof sketch]
%% Suppose $v_i(a)=c$ for all $i \in N$ and $a \in A$. If the underlying
%% player graph is unweighted, then it follows that the value of
%% potential function $\phi$ is bounded above by $2nc + 2|\tau|$ and bounded below
%% by $2nc$ and in each
%% resolution step the value of $\phi$ increases by at least 2.
%% \end{proof}

If we consider the player graph to have weighted edges, then computing
a 2-stable outcome is PLS-complete already for the symmetric \name{}
with uniform valuation and non-negative edge weights.

\begin{restatable}{theorem}{thmPlsComplete}
%% \begin{theorem}
\label{thm:pls-complete}
Finding a 2-stable allocation in a symmetric \name{} with uniform
valuation in which the edge weights in the underlying player graph are
non-negative is PLS-complete.
%Finding a 2-stable allocation for \emph{SHP-UV} is PLS-complete.
%\end{theorem}
\end{restatable}

\begin{proof}[Proof sketch]
%Consider an instance of the symmetric \name{} with uniform valuation. 
Without loss of generality we assume that $v_i(a)=0$ for all $i \in N$
and $a \in A$. The potential function defined in the proof of
Theorem~\ref{thm:sym-potential} essentially reduces to the social
welfare, i.e., $\phi(\pi)=\Sigma_{i \in N} u_i(\pi)=\sw(\pi)$. It can
be verified that computing a 2-stable allocation is in PLS. 
%% (any blocking
%% pair can be found in polynomial time, and a blocking pair resolution
%% is a local improvement).
To show PLS-hardness, we give a \emph{tight reduction} from the
max-cut problem with FLIP neighbourhood \cite{schaffer1991PLS}.

Let $Q=(V,E,\{z_e\}_{e \in E})$ be an instance of the max-cut problem where $V$
is the set of vertices, $E \subseteq V \times V$ is the set of edges and
$\{z_e\}_{e \in E}$ is the set of non-negative edge weights.
%% Let $Q = (V, E_W))$ be the graph for which we need to find a locally
%% optimum max-cut. (where $V$ denotes the vertices, $E_W =
%% \{(a,b,w_{ab})\mid a,b \in A, w_{ab}$ is weight of edge between a and
%% b $\}$ denotes the edges, along with their weights.)
We construct an instance of the symmetric \name{} in which the
underlying player graph has non-negative edge weights as follows. Let
the player graph $G = (N, \tau)$, where $N=\{v_{\txtblue} \mid v \in
V\} \cup \{ v_{\txtred} \mid v \in V\}$. That is, for every vertex $v
\in V$, there are two vertices $v_{\txtblue}$ and $v_{\txtred}$ in
$N$. Thus $|N|=2*|V|$. For each edge $(u,v) \in E$ we add two edges in
$\tau$: $(u_{\txtblue},v_{\txtblue}) \in \tau$ and
$(u_{\txtred},v_{\txtred})\in \tau$ with $w_{u_{\txtblue},
v_{\txtblue}} = w_{u_{\txtred}, v_{\txtred}} = z_{(u,v)}$.  Let
$w_{\mathit{max}} = |N|^2*(\max_{e \in E}z_e - \min_{e \in E}z_e)$.
For every $v \in V$, we also add the edge $(v_{\txtblue},
v_{\txtred}) \in \tau$ with $
w_{v_{\txtblue}, v_{\txtred}} = w_{\mathit{max}}$.
%%  = (|N|^2*(max (0,maxWeight (w \in E_W)) - min(0, minWeight(w \in E_W)))) = w_{max}$. 
The item graph is the complete bipartite graph $H = (A, \lambda)$
where each partition consists of $|V|$ vertices. In other words, let
the two partitions be $A_1$ and $A_2$, with vertices $A_1(1),
A_1(2),\ldots,A_1(|V|)$ and $A_2(1), A_2(2),\ldots,A_2(|V|)$). For an
arbitrary cut $(V_1, \obar{V_1})$, we can construct an
allocation $\pi$ such that $\sw(\pi) =
4*\mathit{cutWeight}(V_1, \obar{V_1}) + 2|V|*w_{max}$, that proves
the result.
%% By
%% assumption, for all $i \in N$ and $a \in A$, $v_i(a) = 0$.
\end{proof}

Since Theorem~\ref{thm:pls-complete} gives a tight PLS reduction from
the max-cut problem, we have the following corollary.

\begin{restatable}{corollary}{corExpPath}
%\begin{corollary}
The standard local search algorithm takes exponential time in the
worst case for the symmetric \name{}.
%\end{corollary}
\end{restatable}

The above result shows that in symmetric \name{}s with uniform
valuations, while a potential function exists, the bound on the
function can be exponential in the encoding of the instance. A natural
question is to ask if the utilities can be replaced by some
bounded integer function for which the local search has the exact same
behaviour. We now show that when the degree of the player graph is
bounded by two, this is indeed possible, resulting in a polynomial
time procedure.

%This raises the natural question: Are there
%restricted classes of the symmetric {\name{}} where a 2-stable outcome
%can be computed efficiently?
%
%We now identify restricted classes in which this is
%indeed possible.
%% We show that this is indeed possible if the degree of
%% each node in the underlying player graph is at most 2. 
%% Simple cycles,
%% for instance, satify this condition.

\begin{restatable}{theorem}{thmSymCycle}
%% \begin{theorem}
%% \label{thm:sym-cycle}
For the symmetric \name{} with uniform valuation, if the underlying
player graph has vertices of degree at most two, a 2-stable allocation
can be computed in polynomial time.
%% and the edge weights are non-negative, 
%% \end{theorem}
\end{restatable}
\begin{proof}[Proof sketch]
Let $c_{i,1} = \max_{j \in i_{\tau}} |w_{j, i}|$ with  
$c_{i,2} = \min_{j \in i_{\tau}} |w_{j,i}|$ if $|i_{\tau}| = 2$ and $c_{i,2} = 0$ otherwise ($\forall i \in N$).
%$c_{i,2} = \begin{cases} \min_{j \in i_{\tau}} |w_{j, i}| & \quad \text{if } |i_{\tau}|=2 \\ 0 & \quad \text {otherwise}$. 
Let $g: N \rightarrow N$ denote an ordering of the players (where
$g(i)$ denotes the rank of $i \in N$ in this ordering) which satisfies
the following condition: if $g(i) < g(j)$ then $c_{i,1} \geq
c_{j,1}$.  For $i \in N$, let $u_i^{\min} = \min_{\pi \in
	\Pi}u_i(\pi)$ and
$Y_i=\{u_i^{\min},u_i^{\min}+c_{i,2},u_i^{\min}+c_{i,1},u_i^{\min}+c_{i,1}+c_{i,2}\}$.
Note that, for all $\pi \in \Pi$, $u_i(\pi) \in Y_i$. For $x \in Y_i$, let
%Consider the function $f_i: Y \to \{0,1,2,3\}$ which scales the possible utility
%values, defined as follows:

$f_i(x) = \begin{cases}
0 & \quad \text{if } x = u_i^{\min} \\
1 & \quad \text{if } x = u_i^{\min}+c_{i,2} (c_{i,2} \leq c_{i,1}) \\
2 & \quad \text{if } x = u_i^{\min}+c_{i,1} \text{ and } c_{i,1} > c_{i,2} \\
3 & \quad \text{if } x = u_i^{\min}+c_{i,1}+c_{i,2} \text{ and } c_{i,2} > 0
\end{cases}
$ \\
It can be shown that $P(\pi) = \sum_{i \in N}
(2n - g(i)) f_i(u_i(\pi))$ is a potential function with an upper bound of $6n^2$ and a step size of at least 1 ensuring computation in polynomial time. 
\end{proof}

The above result implies efficient computation for various classes of
graphs which are often studied, for instance, simple cycles. One
natural question is whether the result can be extended to structures
with small (logarithmic) neighbourhood. Such a restriction on the
neighbourhood graph has interesting consequences in various classes of
strategic form games where equilibrium computation is generally known
to be hard \cite{gottlob_pure_2005}. We now show that for bounded
degree graphs the problem remains PLS-complete.

\begin{restatable}{theorem}{thmSymConstHard}
\label{thm:symConstHard} 
Finding a 2-stable allocation in a symmetric \name{} with uniform valuation in
which the underlying player graph has vertices with degree at most six
is PLS-complete.
\end{restatable}

\begin{proof}[Proof sketch]
\cite{ETmaxCut11} showed that finding a local max-cut with a FLIP 
neighbourhood is PLS-complete even for graphs with degree at most
five. Since our reduction in Theorem~\ref{thm:pls-complete} increases
the degree of each vertex by at most one, the result follows.
\end{proof}
%\begin{corollary}
%For the symmetric \name{} with uniform valuation, if the underlying
%player graph is a \cycle{}, a 2-stable allocation can be computed
%in polynomial time.
%\end{corollary}

While in the analysis above, we consider restrictions on the player
graph, it is also natural to study restrictions on the item graph. A
starting point would be to consider the case when the item graph is a
complete graph. It can be verified that in this situation, the
externalities do not play a crucial role and a core stable allocation
can be computed using the TTC algorithm. A similar observation holds
when the edge set in the item graph is empty ($\lambda = \emptyset$). Another candidate restriction
would be the bipartite graph. However, our
PLS-hardness reduction in Theorem \ref{thm:pls-complete}
constructs a complete bipartite item graph. 
A careful analysis of the potential function in the context of complete
bipartite item graphs provides an upper bound in terms of the size of the partition.
%We now do a more careful
%analysis of the potential function in the context of complete
%bipartite item graphs and relate it to the size of the partition.

%\red{Theorem \ref{thm:symConstHard} establishes the difficulty in efficient
%computation even with restricted player graphs. Thus, it is logical to
%look at restrictions on the item graph. An obvious restriction would a
%complete item graph. However, it is easy to observe that the problem
%reduces to the standard one-sided market problem in that case. Our
%reduction in theorem \ref{thm:pls-complete} relies on the item graph
%being a complete bipartite graph to show PLS-completeness. A careful
%analysis of the potential function shows that computing a 2-stable
%allocation is in XP with respect to the size of the smaller partition.}

\begin{restatable}{theorem}{thmSymBip}
For the symmetric \name{}, if the underlying item graph is a complete
bipartite graph $H = (U, V, U \times V)$ with $U$ and $V$ being the 2
partitions, a 2-stable allocation can
be computed in $O(n^{\text{min}(|U|,|V|)+4})$.
\end{restatable}
\begin{proof}[Proof sketch]
Let $N_U = \{i \mid \pi(i) \in U\}$ and $N_V = \{i \mid \pi(i) \in V\}$. Then
the potential function reduces to $\phi(\pi)=\sum_{i \in N_U}
2u_i(\pi) + \sum_{i \in N_V} 2v_i(\pi)$. Assume $|U| \leq |V|$. For
each assignment $\pi$, we can find the maximum value of $\sum_{i \in
N_V} 2v_i(\pi)$ by finding the maximum weight matching of the
bipartite graph $Q = (V, N_V, V \times N_V)$ with the weight of an
edge $(i, a)$ being $v_i(a)$ (in $O(n^4)$). Since there are $\frac{n!}{(n - |U|)!}$ ways 
of assigning players to $U$, an allocation with the optimal potential value can be computed in
$O(n^{\text{min}(|U|,|V|)+4})$.
\end{proof}

\begin{corollary}
For the symmetric \name{}, if the underlying item graph is a complete
bipartite graph with a constant number of vertices in one of the
partitions, a 2-stable allocation can be computed in polynomial time.
\end{corollary}

There are various interesting graph structures which satisfy the above
restriction, the simplest being the star graph, which is independently important in modelling workplaces with flat hierarchies, such as research labs. 
%\red{The simplest graph satisfying this restriction is the
%star graph, which is independently important in modelling workplaces
%with flat hierarchies, such as research labs.}
%%  where the player graph is a cycle, with
%% non-negative neighbourhood payoffs, irrespective of the initial
%% allocation, any sequence of blocking pair resolutions will converge to
%% a 2-stable allocation with the length of the sequence being
%% polynomially bounded.

\section{Stability in Asymmetric Neighbourhood}
\label{sec:asymmetric}
When the underlying player graph is not symmetric,
Example~\ref{ex:no-2stable} shows that 2-stable outcomes need not
always exists. The negative result already holds for the allocation
problem with uniform valuation where the underlying player graph is a
\cycle{} and the item graphs consists of two disconnected
cliques. This raises the question: What is the complexity of deciding
if an instance of the \name{} has a 2-stable outcome? We show that in
general, this problem is NP-complete. We then identify restrictions
where stable outcomes always exist and can be computed efficiently.

\begin{figure}
%\begin{minipage}{.5\textwidth}
\centering
%\begin{figure*}[htbp]
\tikzstyle{agent}=[circle,draw=black!80,thick, minimum size=2em,scale=0.8]
\begin{tikzpicture}[auto,>=latex',shorten >=1pt,on grid]
\newdimen\R
\R=1.3cm
\newcommand{\llab}[1]{{\small $\{#1\}$}}
%\begin{tikzpicture}[scale=4,cap=round,>=latex]

\node[agent] (1)                    {\footnotesize $C_{1,i}$};
\node[agent] (2) [right= 3 of 1]       {\footnotesize $C_{2,i}$};
\node[agent] (3) [below= 2.5 of 1]       {\footnotesize $C_{3,i}$};
\node[agent] (4) [right= 3 of 3]       {\footnotesize $C_{4,i}$};

%% \foreach \x/\y in {1/2,3/4}{
%% \draw[->] (\x) to node [below] {\footnotesize $-2d$} (\y);
%% }

\draw[->] (1) to node [above] {\footnotesize $d$} (2);
\draw[->] (3) to node [below] {\footnotesize $d$} (4);

%% \draw[->] (1) to node [left] {\footnotesize $-2d$} (3);
%% \draw[->] (2) to node [right] {\footnotesize $-2d$} (4);

\draw[->, bend right=50] (2) to  node [above] {\footnotesize $d-e$} (1);
\draw[->, bend left = 50] (4) to  node [below] {\footnotesize $d-e$} (3);

\draw[<->, bend left=65] (3) to  node [left] {\footnotesize $-2d$} (1);
\draw[<->, bend right=65] (4) to  node [right] {\footnotesize $-2d$} (2);

\draw[->, bend left=20] (1) to  node [pos=.8,right] {\footnotesize $d-e$} (4);
\draw[->, bend left=20] (2) to  node [pos=.1,right] {\footnotesize $d$} (3);

\draw[->, bend left=20] (3) to  node [pos=.2,left] {\footnotesize $d-e$} (2);
\draw[->, bend left=20] (4) to  node [pos=.9,left] {\footnotesize $d$} (1);

\end{tikzpicture}
\caption{A gadget
\label{fig:clause}
%\vspace*{-0.5em}
}
\end{figure}
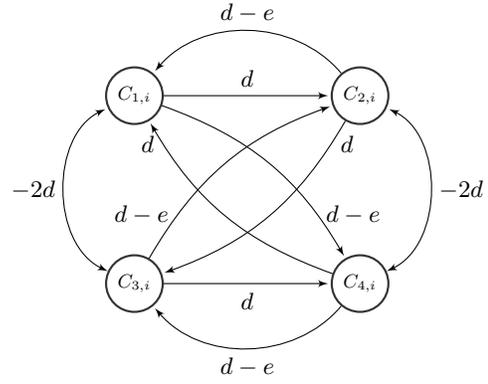

%% Check if a game has a 2-stable outcome
\begin{restatable}{theorem}{thmCheckTwoStableNP}
%% \begin{theorem}
%% \label{thm:checkstable-NP}
The problem of deciding if an instance of the \name{} has a 2-stable
allocation is NP-complete.  
%% Given an instance of \name{} deciding if
%% there exists a two stable allocation in $\alloc$ is NP-complete.
%% is 2-STABLE-HP is NP-complete.
%% \end{theorem}
\end{restatable}
\begin{proof}[Proof sketch]
Given an allocation $\pi$, deciding whether $\pi$ is a 2-stable
allocation can be done in polynomial time. It suffices to check if
there is a blocking pair in $\pi$. Thus the above problem is in NP. To
show hardness, we give a reduction from 3-SAT.
Let the 3-SAT instance have $q$ variables ($\{a_1,a_2,\ldots,a_q\}$) and $m$ clauses ($\{c_1,c_2,\ldots, c_m\}$). We 
create an instance of the \name{} with $|N| = 4m+2q$, with the players 
being differentiated into 6 types with
$N = C_1 \cup C_2 \cup C_3 \cup C_4 \cup S_1 \cup S_2$, where $\forall i \in \{1,2,3,4\}$ $|C_i| = m$
and $\forall j \in \{1,2\}$ $|S_j| = q$.
The key idea is to set up item valuations and the player neighbourhood 
structure such that a stable allocation can only exist if each player is 
assigned an item from the item set corresponding to its type(s) \textit{and} 
the 3-SAT instance is satisfiable. We set up the following constants to aid
our explanation: $b = (2q + 3m + 20)d = (2q + 3m + 20)^{2}e = (2q + 3m + 20)^{3}f$ with $f = 1$.
There are 4 players corresponding to 
each clause and 2 players corresponding to each variable. For each clause $c_i$, 
the players corresponding to it ($c_{i,1}$,$c_{i,2}$, $c_{i,3}$ and $c_{i,4}$)
are connected as shown in 
Figure \ref{fig:clause}, with $w_{c_{i,2}, c_{i,1}} = w_{c_{i,4}, c_{i,3}} = d - e$,
$w_{c_{i,1}, c_{i,2}} = w_{c_{i,3}, c_{i,4}} = d$, 
$w_{c_{i,4}, c_{i,1}} = w_{c_{i,2}, c_{i,3}} = d$, 
$w_{c_{i,1}, c_{i,4}} = w_{c_{i,3}, c_{i,2}} = d - e$ and 
$w_{c_{i,3}, c_{i,1}} = w_{c_{i,1}, c_{i,3}} = w_{c_{i,2}, c_{i,4}} = w_{c_{i,4}, c_{i,2}} = -2d$. 
For each variable $a_j$, there are two
players $s_{j,1}$ and $s_{j,2}$ corresponding to the positive literal $a_j$
and the negative literal $\neg a_j$ respectively. 
The edge connecting these two players has a large negative weight
($w_{s_{j,1},s_{j,2}} = w_{s_{j,2},s_{j,1}} = -d$). 
For each clause $c_i$ where a positive literal $a_j$ appears $x$ times
and the corresponding negative literal $\neg a_j$ appears $y$ times,
$w_{s_{j,1}, c_{i,1}} = (x - y)e$ and $w_{s_{j,2}, c_{i,1}} = (y - x)e$.
For example, suppose $c_i = a_{t} \vee \neg a_{u} \vee a_{v}$. ($t \neq u \neq v$) Then, 
$w_{s_{t,1}, c_{i,1}} = w_{s_{u,2}, c_{i,1}} = w_{s_{v,1}, c_{i,1}} = e$ and 
$w_{s_{t,2}, c_{i,1}} = w_{s_{u,1}, c_{i,1}} = w_{s_{v,2}, c_{i,1}} = -e$.

	%\begin{tabular}{|l|l|l|l|l|l|}
	\begin{table}
        \centering
		%\label{FTable} 
		\resizebox{0.45\textwidth}{!}{%
			%\begin{tabular}{|c|@{}c@{}|@{}c@{}|@{}c@{}|@{}c@{}|@{}c@{}|}
			\begin{tabular}{|l|l|l|l|l|l|}
				%\begin{tabular}{|c@{}|c@{}|c@{}|c@{}|c@{}|c@{}|}
				\hline
				& $p \in S_1 \cup S_2$&$p \in C_1$&$p \in C_2$&$p \in C_3$&$p \in C_4$ \\
				\hline
				$t \in A_{S_1}$ & $b$ & 0 & 0 & 0 & 0 \\
				\hline
				$t \in A_{S_2}$ & $b$ & 0 & 0 & 0 & 0 \\
				\hline
				$t \in A_{C_1}$ & 0 & $b+4e-f$ & $b-6d$ & $b$ & $b-6d$ \\
				\hline
				$t \in A_{C_2}$ & 0 & $b-6d$ & $b$ & $b-6d$ & $b - f$ \\
				\hline
				$t \in A_{C_3}$ & 0 & $b-f$ & $b-6d$ & $b$ & $b-6d$ \\
				\hline
				$t \in A_{C_4}$ & 0 & $b-6d$ & $b-f$ & $b-6d$ & $b$ \\
				% cell4 & cell5 & cell6 \\ 
				% cell7 & cell8 & cell9 \\ 
				\hline
				
			\end{tabular} } \end{table}

The item graph is $A = A_{C_1} \cup A_{C_2} \cup A_{C_3} \cup
A_{C_4} \cup A_{S_1} \cup A_{S_2}$, with $|A_{C_i}| = m$ and
$|A_{S_j}| = q$, with each of these six subsets being cliques.
Additionally, there is complete connection between $A_{S_2}$ and
$A_{C_1}$, $A_{C_1}$ and $A_{C_2}$ and $A_{C_3}$ and $A_{C_4}$.  The
table above contains $v_p(t)$, for each player $p$ and item $t$.
%the valuation of item $t$ for player $p$.
\end{proof}

Thus, deciding the existence of a 2-stable allocation in an
instance of the \name{} is NP-complete. It is natural to try to identify
restricted classes in which stable allocations are guaranteed to
exist and classes where such allocations can be computed
efficiently. A natural restriction to consider is a hierarchical
influence structure, which is present in several organizations. In our
framework, such a structure can be modelled by restricting the player
graph to a directed acyclic graph (DAG). We now show that when the
player graph is a DAG, a core stable allocation is guaranteed to exist
and such an allocation can be computed in polynomial time.

\begin{restatable}{theorem}{thmTwoStableDAG}
% Consider an instance $M$ of \name{} where the underlying player graph
% $\pgraph$ is a DAG. A 2-stable allocation is guaranteed to exist in
% $M$ and it can be computed in polynomial time.
Consider an instance $\alloc$ of the \name{} where the underlying player graph
$\pgraph$ is a DAG. A core stable allocation is guaranteed to exist in
$\alloc$ and it can be computed in polynomial time.
%% \end{theorem}
\end{restatable}

\begin{restatable}{theorem}{thmCycleOneConnectedNode}
\label{thm:CycleOneConnectedNode}
Consider an instance $\alloc = (\pgraph,\igraph,(v_i)_{i \in N})$ of
the \name{} with uniform valuation. If $\alloc$ satisfies the following
conditions, then a 2-stable allocation always exists and it can be
computed in polynomial time.
\begin{itemize}
\item $G$ is a \cycle{}  and $\exists i,j \in N$ such that $(j,i)
  \in \tau$ and $w_{j,i} >0$.
\item there exists $a \in A$ such that for all $a' \in A - \{a\}$, $(a,a') \in
  \lambda$. 
\end{itemize}
%% If the player graph is a cycle with the neighbourhood valuation being
%% non-negative for at least one player, the item valuations being
%% uniform, and the item graph has an item connected to all the other
%% items, a 2-stable allocation can be computed in polynomial time.
%% \end{theorem}
\end{restatable}

\begin{restatable}{theorem}{thmCycleOneDegreeItemCoreStable}
%% \begin{theorem}
\label{thmCycleOneDegreeItemCoreStable}
Consider an instance $\alloc = (\pgraph,\igraph,(v_i)_{i \in N})$ of the
\name{} with uniform valuation. If $\pgraph$ is a \cycle{} with 
positive connection weights $w$ and $\igraph$ is connected and has at least 
1 node with degree 1, then a core stable allocation always exists and it can be
computed in polynomial time.
%\end{theorem}
\end{restatable}

% The above proof shows that the 2-stable allocation constructed in
% Theorem~\ref{thmCycleOneDegreeItemTwoStable} is also
% core stable. 
%% The above proof shows the 2-stable allocation coinciding with the core stable allocation.
Note that it is not the case that every 2-stable allocation is also a
core stable allocation in this restricted setting (as specified in
Theorem \ref{thmCycleOneDegreeItemCoreStable}). Consider an instance of
the \name{} with uniform valuation where $N=\{1,\ldots,6\}$ forms a
\cycle{} and the item graph is as shown in
Figure~\ref{fig:2stableNotCore}. The numbers labelling the nodes of
the item graph denote the players to which the items are assigned.
%% with the items being labelled with the players to which
%% they are assigned.
Consider the allocation $\pi$ where
$\pi(i)=A_i$ for $i \in \{2,4,6\}$ and $\pi_i=B_i$ for $i \in
\{1,3,5\}$ (also labelled in Figure~\ref{fig:2stableNotCore}). It can
be verified that $\pi$ is 2-stable but not core stable due to the
existence of the blocking coalition $X=\{1,3,5\}$. In fact, at $\pi$
there is a unique blocking coalition given by $X$. The resolution of
the blocking coalition generates an allocation $\pi'$ where
$\pi'(1)=B_5, \pi'(3)=B_3$ and $\pi'(5)=B_5$. Now $X'=\{2,4,6\}$ forms
a blocking coalition in $\pi'$. It can also be verified that starting
at $\pi$ there is a unique sequence of improvement steps resolving
blocking coalitions which results in an infinite coalition improvement
path.

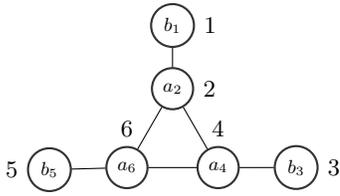
\begin{figure}
\centering
\tikzstyle{agent}=[circle,draw=black!80,thick, minimum size=1em,scale=0.8]
\begin{tikzpicture}[auto,>=latex',on grid]
%% shorten >=1pt,
\newdimen\R
\R=.7cm
\newcommand{\llab}[1]{{\footnotesize $#1$}}
%\begin{tikzpicture}[scale=4,cap=round,>=latex]
\draw (90: \R) node[agent,label=right:{\llab{2}}] (2) {\footnotesize $a_2$};
\draw (90-120: \R) node[agent,label=above:{\llab{4}}] (4) {\footnotesize $a_4$};
\draw (90-240: \R) node[agent,label=above:{\llab{6}}] (6) {\footnotesize $a_6$};
\draw (90: 2.2*\R) node[agent,label=right:{\llab{1}}] (1) {\footnotesize $b_1$};
%% \draw (90-120: 2.4*\R) node[agent,label=right:{\llab{3}}] (3) {\footnotesize $b_3$};
%% \draw (90-240: 2.4*\R) node[agent,label=left:{\llab{5}}] (5) {\footnotesize $b_5$};
%\node[agent] (3) [right=1 of 4]       {\footnotesize $b_3$};
%\node[agent] (5) [left=1 of 6]       {\footnotesize $b_5$};
\draw (90-102: 2.4*\R) node[agent,label=right:{\llab{3}}] (3) {\footnotesize $b_3$};
\draw (90-257: 2.4*\R) node[agent,label=left:{\llab{5}}] (5) {\footnotesize $b_5$};
\foreach \x/\y in {2/4,4/6,6/2,1/2,3/4,5/6} {
    \draw[-] (\x) to (\y);    
}
\end{tikzpicture}
\caption{An item graph and an allocation
\label{fig:2stableNotCore}
%\vspace*{-0.5em}
}
\end{figure}

\section{Envy-Freeness}
\label{sec:envy}
%\red{Motivation}
Envy-freeness is a natural and well-studied notion of fairness in
resource allocation.  In the cake-cutting problem, the existence of a complete,
envy-free allocation is guaranteed for all agent
valuations \cite{alon1987ccef}. In the indivisible goods setting, we
have hardness results \cite{bouveret2008indiv} for computing the
existence of a complete, envy-free allocation (Note: An allocation is complete if all of the item(s) available has (have) been assigned to the players).  Recent works on resource
allocation with a graphical component (\cite{CEM17}, \cite{FD-graph17}, \cite{abebe2017fair} and \cite{beynier2018local}) have examined envy-freeness and its variants
in great detail.  
%% Thus, it is natural to examine envy-freeness in our modified
%% resource allocation problem.
Typically, the notion of envy is defined with respect to a bundle of
items possessed by some other player. In our model, since a player's
utility is dependent on other players' allocations, we adopt a
slightly modified notion of envy-freeness, 
% which is very similar to
% %the notion of
% %Single reference or should there be a reference to Velez
% \textit{swap envy-freeness} defined in \cite{BPZ13},
which is defined as follows.
An allocation $\pi$ is envy-free if there does not exist
$i,j$ such that $u_i(\pi) < u_i(\pi')$ where $\pi'(i) = \pi(j)$,
$\pi'(j) = \pi(i)$ and for all $k \in N -\{i,j\}$, $\pi(k) = \pi'(k)$.
The above notion is very similar to the notion of \textit{swap envy-freeness} defined
in \cite{BPZ13}. 
Note that, in a \name{}, every allocation is complete by definition. 
%$f^{-1}(i) = \pi^{-1}(j)$, $f^{-1}(j) = \pi^{-1}(i)$ and $f^{-1}(k) = \pi^{-1}(k)$ $\forall k \in N -\{i,j\}$. 
%%
A natural question is to ask whether it is possible to decide the existence of an  envy-free
allocation given an instance of the \name{}.
We show that even for a \name{} with uniform valuation, 
where the underlying player graph is a cycle, this problem is NP-complete.
 
\begin{restatable}{theorem}{CycleDirectedEF}
%% \begin{theorem}
%% \label{CycleDirectedEF}
For an instance of the \name{} with uniform valuation where the
underlying player graph is an unweighted \cycle{},
deciding if there exists an envy-free allocation 
%in $\alloc$ 
is NP-complete. 
%%  Given an instance of the
%% \name{}, \emph{COMPLETE-EF-CYCLE-HP} is NP-complete
%% \end{theorem}
\end{restatable}

It is not difficult to show that deciding the existence of a
complete, envy-free allocation in an instance of the symmetric \name{}
is NP-complete. Note that, deciding the existence of a complete, envy-free
assignment in an allocation problem is known to be
NP-complete \cite{bouveret2008indiv}.

% It can also be shown that deciding for the existence of an envy-free allocation
% in an instance of the symmetric \name{} is NP-complete. However, since finding
% the existence of an envy-free allocation is NP-complete
% \cite{bouveret2008indiv}, we skip the proof here.
% \red{Add the comment on the symmetric case}
%TODO: 

\section{Conclusions}
%\section{Summary and Outlook}
\label{sec:conclusions}
In this paper we studied the problem of allocating indivisible items
to a set of players where agents have cardinal non-transferable
utilities associated with each allocation. We extended the one-sided
market model to the network setting where agents' utilities depend on
their neighbourhood externalities that are item specific and pairwise
separable.
We show that unlike in the case of one-sided markets, 2-stable
allocations may not always exist.
When the underlying neighbourhood
structure is symmetric, a 2-stable allocation is guaranteed to exist
although computing such an allocation is PLS-complete (already for
player graphs of degree 6). We provide a polynomial time procedure to
compute a 2-stable allocation when the degree of the player graph is
bounded by two. An interesting question is to see if this result can
be extended to player graphs of degree three using the technique
of \textit{local linear programs} as done for the local max-cut
problem \cite{localmaxcut-cubic}. Another natural question is the
existence of core stable outcomes.
While we believe that a core-stable allocation always exists in the
symmetric setting, so far, we have been unable to prove this
result. In case a core-stable outcome is not guaranteed to exist, it
would be useful to find the maximum value of c such that c-stable
outcomes exist.

There are several ways to extend the model. Allocations that allow
subsets of items to be assigned to players is an obvious choice. We
could also consider externalities which are not necessarily pairwise
separable. It would be interesting to see whether the existence
results continue to hold in these extended settings. It would also be
interesting to identify more general classes of neighbourhood
structures in which stable allocations are guaranteed to exist and
where such an allocation can be efficiently computed.

\section*{Acknowledgements}
%% \noindent {\bf Acknowledgements.}
We thank the reviewers for their
useful comments. Sunil Simon was partially supported by grant
MTR/2018/001244.

\bibliographystyle{named}
\bibliography{ref-allocation}  % put name of your .bib file here

%% \newpage
\appendix
\section*{Appendix}

%% Example no 2 stable

\noindent{\bf Example~\ref{ex:no-2stable}.}  Let the set of players
$N=\{1,\ldots,6\}$ and $A=\{A,B,C,D,E,F\}$. Suppose the player graph
$\pgraph=(N,\tau)$ forms a \cycle{} on $N$ consisting of the edges
$(6,1) \in \tau$ and $(i,i+1) \in \tau$ for all $i:1\leq i \leq
5$. For every edge $(i,j) \in \tau$, let $w_{i,j}=c$ for some constant
$c$. We also assume uniform valuation $v_i(a)=0$ for all $i \in N$ and
$a \in A$. Let the item graph be the structure given in
Figure~\ref{fig:no-2stable}. 
%% Let the player graph consist of 6 players ($\{1,2,3,4,5,6\}$),
%%   with $w_{i+1}(i) = c$ ($c > 0$) $\forall i \in \{1,2,3,4,5\}$ and
%%   $w_{1}(6) = c$. Additionally, the item graph consists of 6 items
%%   ($\{A,B,C,D,E,F\}$), with $v_i(t) = 0$ $\forall i \in
%%   \{1,2,3,4,5,6\} \forall t \in \{A,B,C,D,E,F\}$. The item graph
%%   consists of two 3-cliques, $\{A,B,C\}$ and $\{D,E,F\}$
%%   respectively.
Since the item graph consists of two 3-cliques, it is sufficient to
divide players into groups of 3, assign them to any of the 3-cliques,
and check if, for any group of 3, we have a 2-stable allocation. Below
we provide the list of all such allocations and underline the blocking
pair: $((\underline{1},2,3), (\underline{4},5,6))$,
$((\underline{1},2,4),(\underline{3},5,6))$,
$((\underline{1},2,5),(\underline{3},4,6))$,
$((1,2,\underline{6}),(\underline{3},4,5))$,
$((\underline{1},3,4),(2,\underline{5},6))$,
$((\underline{1},3,5),(2,\underline{4},6))$,
$((1,3,\underline{6}),(\underline{2},4,5))$,
$((1,\underline{4},5),(\underline{2},3,6))$,
$((1,\underline{4},6),(\underline{2},3,5))$,
$((1,\underline{5},6),(\underline{2},3,4))$.
%%  following list details the blocking pairs
%% for each of the assignments:
%% \begin{itemize}
%% \item ((1,2,3),(4,5,6)): (1,4) forms a blocking pair.
%% \item ((1,2,4),(3,5,6)): (1,3) forms a blocking pair.
%% \item ((1,2,5),(3,4,6)): (1,3) forms a blocking pair.
%% \item ((1,2,6),(3,4,5)): (3,6) forms a blocking pair.
%% \item ((1,3,4),(2,5,6)): (1,5) forms a blocking pair.
%% \item ((1,3,5),(2,4,6)): (1,4) forms a blocking pair.
%% \item ((1,3,6),(2,4,5)): (2,6) forms a blocking pair.
%% \item ((1,4,5),(2,3,6)): (2,4) forms a blocking pair.
%% \item ((1,4,6),(2,3,5)): (2,4) forms a blocking pair.
%% \item ((1,5,6),(2,3,4)): (2,5) forms a blocking pair.
%% \end{itemize}

%% Stability in symmetric neighbourhood

\thmSymPotential*
\begin{proof}
%% We can argue that the following function acts as a potential
%% function \cite{MS96} in the symmetric \name{}: $\phi(\pi)
%% = \sum_{i \in N} (v_i(\pi) + u_i(\pi))$.

We can argue that the following function acts as a potential
function in the symmetric \name{}: $\phi(\pi)
 = \sum_{i \in N} (v_i(\pi) + u_i(\pi))$.

Consider an arbitrary allocation $\pi$ which is not 2-stable and let
$\pi \to_{p_1,p_2} \pi'$.  By the definition of a blocking pair, for $i
\in \{1,2\}$, $v_{p_i}(\pi') + r_{p_i}(\pi') > v_{p_i}(\pi) +
r_{p_i}(\pi)$.  Therefore, we have
%\begin{align*}
%& \phi(\pi') - \phi(\pi) \\
%&= \sum_{p \in N}(u_p(\pi') + v_p(\pi') - (u_p(\pi) + v_p(\pi))) \\
%&= \sum_{i \in \{1,2\}} ((v_{p_i}(\pi') - v_{p_i}(\pi)) + 2(r_{p_i}(\pi') -
%r_{p_i}(\pi))) + (v_{p_i}(\pi') - v_{p_i}(\pi)) \\
%&= 2\sum_{i \in
%	\{1,2\}} (v_{p_i}(\pi') + r_{p_i}(\pi')) - (v_{p_i}(\pi) +
%r_{p_i}(\pi)) > 0 
%\end{align*} 
$\phi(\pi') - \phi(\pi) = \sum_{p
  \in N}(u_p(\pi') + v_p(\pi') - (u_p(\pi) + v_p(\pi)))$ $= \sum_{i
  \in \{1,2\}} ((v_{p_i}(\pi') - v_{p_i}(\pi)) + 2(r_{p_i}(\pi') -
r_{p_i}(\pi))) + (v_{p_i}(\pi') - v_{p_i}(\pi))$ $= 2\sum_{i \in
  \{1,2\}} (v_{p_i}(\pi') + r_{p_i}(\pi')) - (v_{p_i}(\pi) +
r_{p_i}(\pi))$ $> 0$
\end{proof}

%% PLS completeness for symmetric allocation problem

\thmPlsComplete*   
\begin{proof}
Consider an instance of the \name{} with uniform valuation. Without loss
of generality we assume that $v_i(a)=0$ for all $i \in N$ and $a \in
A$. In this case, the potential function defined in the proof of
Theorem~\ref{thm:sym-potential} essentially reduces to the social
welfare, i.e., $\phi(\pi)=\Sigma_{i \in N}
u_i(\pi)=\sw(\pi)$. Note that, an arbitrary allocation can be computed in 
polynomial time. Given an arbitrary allocation $\pi$, we
can compute the social welfare $\phi(\pi)$ in polynomial time. We 
can also check if $\pi$ is 2-stable in polynomial time. 
%An allocation which corresponds to the local maximum of $\phi$ is 2-stable. 
Thus, the problem of finding a 2-stable allocation is in PLS. To show
PLS-hardness, we give a \emph{tight reduction} from the PLS-complete max-cut
problem with FLIP neighbourhood \cite{schaffer1991PLS}.

%% We will now show a reduction from max-cut (which is
%% PLS-complete) to 2-stable allocation finding in \emph{SHP-UV}.

Let $Q=(V,E,\{z_e\}_{e \in E})$ be an instance of the max-cut problem where $V$
is the set of vertices, $E \subseteq V \times V$ is the set of edges and
$\{z_e\}_{e \in E}$ is the set of non-negative edge weights.
%% Let $Q = (V, E_W))$ be the graph for which we need to find a locally
%% optimum max-cut. (where $V$ denotes the vertices, $E_W =
%% \{(a,b,w_{ab})\mid a,b \in A, w_{ab}$ is weight of edge between a and
%% b $\}$ denotes the edges, along with their weights.)
We construct an instance of the symmetric \name{} in which the
underlying player graph has non-negative edge weights as follows. Let
the player graph $G = (N, \tau)$, where $N=\{v_{\txtblue} \mid v \in
V\} \cup \{ v_{\txtred} \mid v \in V\}$. That is, for every vertex $v
\in V$, there are two vertices $v_{\txtblue}$ and $v_{\txtred}$ in
$N$. Thus $|N|=2*|V|$. For each edge $(u,v) \in E$ we add two edges in
$\tau$: $(u_{\txtblue},v_{\txtblue}) \in \tau$ and
$(u_{\txtred},v_{\txtred})\in \tau$ with
$w_{u_{\txtblue}, v_{\txtblue}} = w_{u_{\txtred}, v_{\txtred}} = z_{(u,v)}$.
Let $w_{\mathit{max}} = |N|^2*(\max_{e \in E}z_e - \min_{e \in E}z_e)$.
For every vertex $v \in V$, we also add the edge
$(v_{\txtblue}, v_{\txtred}) \in \tau$ with
$w_{v_{\txtblue}, v_{\txtred}} =
w_{\mathit{max}}$.
%%  = (|N|^2*(max (0,maxWeight (w \in E_W)) - min(0, minWeight(w \in E_W)))) = w_{max}$. 
The item graph is the complete bipartite graph $H = (A, \lambda)$
where each partition consists of $|V|$ vertices. In other words, let
the two partitions be $A_1$ and $A_2$, with vertices $A_1(1),
A_1(2),\ldots,A_1(|V|)$ and $A_2(1), A_2(2),\ldots,A_2(|V|)$). By assumption,
for all $i \in N$ and $a \in A$, $v_i(a) = 0$.
%% with
%% $|A|=2|V|$ and $|V|$ vertices in each of the two partitions. (Let the
%% two partitions be $A_1$ and $A_2$, with vertices $A_1(1),
%% A_1(2),\ldots,A_1(|V|)$ and $A_2(1), A_2(2),\ldots,A_2(|V|)$). Additionally,
%% $v_i(a) = 0 \forall i \in N \forall a \in A$.

Consider any allocation $\pi$ which is a local optimum in an instance
of the \name{} constructed above. We argue that it satisfies the following
condition: For every pair of players $(v_{\txtblue}, v_{\txtred})$, the items
allocated to these players under $\pi$ will not be in the same
partition. Suppose not, assume that there exists some local optimum
such that there exists nodes $u_{\txtblue}$ and $u_{\txtred}$ such that
$\pi(u_{\txtblue})$ and $\pi(u_{\txtred})$ are in the same partition. Without
loss of generality, let us assume that both the items are in
partition $A_1$. Then there must exist some $v$ such that $v_{\txtblue}$
and $v_{\txtred}$ both possess items in $A_2$. But, $(v_{\txtblue},
u_{\txtblue})$, $(v_{\txtblue}, u_{\txtred})$, $(v_{\txtred}, u_{\txtred})$ and $(v_{\txtred},
u_{\txtblue})$ all form blocking pairs in this case. This gives us a
contradiction to the optimality of $\pi$.

By the definition of a tight reduction (\cite{CZ04overviewPls}, Definition 5.2), we need to be able to find a subset $S$ of allocations which satisfies the conditions below:
\begin{itemize}
\item {$S$ contains all 2-stable allocations of this instance of the \name{}.}
\item {For every cut $(V_1, \obar{V_1})$, it should be possible to construct, in polynomial time,  an allocation $\pi \in S$ such that $\pi$ maps to $(V_1, \obar{V_1})$.}
\item {For any pair of allocations $\pi_1, \pi_t \in S$ such that
  there is an improvement path $\pi_1, \pi_2, \ldots,\pi_{t-1}, \pi_t$
  with the allocations $\pi_2,\ldots ,\pi_{t-1} \notin S$, the cuts $(V_1,
  \obar{V_1})$ and $(V_t, \obar{V_t})$ (corresponding to the
  allocations $\pi_1$ and $\pi_t$ respectively) must either be the
  same or there must be an improvement step from $(V_1, \obar{V_1})$
  to $(V_t, \obar{V_t})$.}
\end{itemize}
 We argue that the following set of allocations satisfies these conditions:
$S = \{\pi| \nexists v \in V, \pi(v_{\txtblue}) \in A_i \text{ and } \pi({v_{\txtred}}) \in A_i \forall i \in \{1,2\} \}$.  

By the necessary condition for local optima established above, all local optima must belong to $S$. Note that, we can construct an allocation corresponding to any cut as follows:
Let $(V_1, \obar{V_1})$ be an arbitrary cut. $\pi$ is an allocation
corresponding to $(V_1, \obar{V_1})$ if $\forall v \in V_1$, $\pi(v_{\txtblue}) \in A_1$
and $\forall v \in \obar{V_1}$, $\pi(v_{\txtblue}) \in A_2$. Note that, 
the social welfare $\sw(\pi) = 4*\mathit{cutWeight}(V_1, \obar{V_1}) + 2|V|*w_{max}$. Such an allocation can clearly be constructed in polynomial time. 

Note that there can be no directed edges from $\pi \in S$ to $\pi' \in
\Pi - S$ in the transition graph of an instance of \name{} (since
$\sw(\pi') < 2|V|*w_{max}$). Thus, the third condition for a tight
reduction is trivially satisfied, which completes our proof.
%Now consider any allocation $\pi$ which is a local optimum. Let $V_1 =
%\{v \mid \pi(v_{\txtblue}) \in A_1\}$. As argued above, since $\pi$ is a 
%local optimum, this implies that $\forall v \in V_1$ $\pi(v_{\txtred}) \in A_2$. 
%%% exactly $|V_{1}|$ items from $A_{1}$ be
%%% assigned to $V_{1_{\txtblue}} \subseteq V_{\txtblue}$. Clearly, the \emph{SW}
%%% = $2*cutWeight(V_1, V - V_1) + |V|*w_{max}$.
%The corresponding cut in the instance of the max-cut problem would be $(V_1, \obar{V_1})$. 
%%Let $(V_1, \obar{V_1})$ be the corresponding cut in the instance of
%%the max-cut problem. 
%It can be verified that $\sw(\pi) =
%4*\mathit{cutWeight}(V_1, \obar{V_1}) + 2|V|*w_{max}$.  We need to
%argue that $(V_1,\obar{V_1})$ is a local optimum. Suppose not, 
%%% Then
%%% there exists some local improvement for this instance.
%without loss of generality, let us assume that the local improvement
%is $(V_1 - \{u\}, \obar{V_1} \cup \{u\})$. % \red{Expand the text below}\\
%Therefore, $\mathit{cutWeight}(V_1 - \{u\}, (\obar{V_1}) \cup \{u\}) > \mathit{cutWeight}(V_1, \obar{V_1})$.
%We construct the allocation $\pi'$ corresponding to $(V_1 - \{u\}, \obar{V_1} \cup \{u\})$. But,
%$\sw(\pi') = 4*\mathit{cutWeight}(V_1 - \{u\}, \obar{V_1} \cup \{u\}) + 2|V|*w_{max} > \sw(\pi)$
%and $\pi \to_{u_{\txtblue}, u_{\txtred}} \pi'$. This gives a contradiction to the optimality of $\pi$
%which completes our proof. 
\end{proof}

%The standard local search algorithm (\cite{CZ04overviewPls}) starts
The standard local search algorithm starts from an initial feasible
solution and moves to better neighbours until it reaches a local
optimum. In the case of the \name{}, this implies starting with an
initial allocation and successively performing blocking pair
resolutions until we reach a local optimum. We show that, for the
\name{}, the standard local search algorithm takes exponential time in
the worst case, irrespective of the blocking pair selection rule used.

\corExpPath*
\begin{proof}
Note that, by theorem 5.15 of \cite{schaffer1991PLS}, the standard
local search algorithm takes exponential time in the worst case for
the max-cut problem. Since we have a tight reduction from the max-cut
problem to the \name{}, by lemma 3.3 of \cite{schaffer1991PLS}, the
standard local search algorithm takes exponential time in the worst
case for the \name{}.
%Since we have a tight reduction from the max-cut problem to the \name{}
%and the standard local search algorithm takes exponential time in the worst case for the max-cut
%problem (\cite{schaffer1991PLS}, Theorem 5.15), by lemma 3.3 of \cite{schaffer1991PLS}, 
\end{proof}

\thmSymCycle*

\begin{proof}
Without loss of generality, assume that $v_i(a) = 0$ for all $i \in N$
and $a \in A$.
%$v_i(a) = 0$.  
Since all the vertices of the underlying player graph $G$ have degree
at most 2, for all $i \in N$, $|i_\tau| \leq 2$. Let $n = |N|$.  In
order to fix notation, let us assume that $i_\tau=\{i_l, i_r\}$ if
$|i_\tau| = 2$ and $i_\tau=\{i_l\}$ if $|i_\tau| = 1$. 
Let $c_{i,1} = \max_{j \in i_{\tau}} |w_{j, i}|$ with  
$c_{i,2} = \min_{j \in i_{\tau}} |w_{j,i}|$ if $|i_{\tau}| = 2$ and $c_{i,2} = 0$ otherwise.
%Let $c_{i,1} =
%\max(|w_{i_l,i}|,|w_{i_r,i}|, 0)$ and $c_{i,2} =
%\min(|w_{i_l,i}|,|w_{i_r,i}|, 0)$.
	
%Let $g: N \rightarrow N$ denote an ordering of the players (where
%$g(i)$ denotes the rank of $i \in N$ in this ordering) which satisfies
%the following condition: if $g(i) < g(j)$ then $|c_{i,1}| \geq
%|c_{j,1}|$.  For $i \in N$, let $u_i^{\min} = \min_{\pi \in
%  \Pi}u_i(\pi)$ and
%$Y_i=\{u_i^{\min},u_i^{\min}+|c_{i,2}|,u_i^{\min}+|c_{i,1}|,u_i^{\min}+|c_{i,1}|+|c_{i,2}|\}$.
%Note that, for all $\pi \in \Pi$, $u_i(\pi) \in Y_i$. Consider the
%function $f_i: Y \to \{0,1,2,3\}$ which scales the possible utility
%values, defined as follows:
%	
%	$f_i(x) = \begin{cases}
%	0 & \quad \text{if } x = u_i^{\min} \\
%	1 & \quad \text{if } x = u_i^{\min}+|c_{i,2}| (|c_{i,2}| \leq |c_{i,1}|) \\
%	2 & \quad \text{if } x = u_i^{\min}+|c_{i,1}| \text{ and } |c_{i,1}| > |c_{i,2}| \\
%	3 & \quad \text{if } x = u_i^{\min}+|c_{i,1}|+|c_{i,2}| \text{ and } |c_{i,2}| > 0
%	\end{cases}
%	$

Let $g: N \rightarrow N$ denote an ordering of the players (where
$g(i)$ denotes the rank of $i \in N$ in this ordering) which satisfies
the following condition: if $g(i) < g(j)$ then $c_{i,1} \geq
c_{j,1}$.  For $i \in N$, let $u_i^{\min} = \min_{\pi \in
	\Pi}u_i(\pi)$ and
$Y_i=\{u_i^{\min},u_i^{\min}+c_{i,2},u_i^{\min}+c_{i,1},u_i^{\min}+c_{i,1}+c_{i,2}\}$.
Note that, for all $\pi \in \Pi$, $u_i(\pi) \in Y_i$. For $x \in Y_i$, let
%Consider the function $f_i: Y \to \{0,1,2,3\}$ which scales the possible utility
%values, defined as follows:

$f_i(x) = \begin{cases}
0 & \quad \text{if } x = u_i^{\min} \\
1 & \quad \text{if } x = u_i^{\min}+c_{i,2} (c_{i,2} \leq c_{i,1}) \\
2 & \quad \text{if } x = u_i^{\min}+c_{i,1} \text{ and } c_{i,1} > c_{i,2} \\
3 & \quad \text{if } x = u_i^{\min}+c_{i,1}+c_{i,2} \text{ and } c_{i,2} > 0
\end{cases}
$

Let $\sigma_i(\pi)= (2n - g(i)) f_{i}(u_{i}(\pi))$ and $P: \Pi \to
\mathbb{Z}_{\geq 0}$ be defined as: $P(\pi) = \sum_{i \in N}
\sigma_i(\pi)$.  Note that for any $\pi \in \Pi$, $P(\pi)$ is at most
$6n^2$. Let $(\pi,\pi')$ be any blocking pair resolution involving
players $(i,j)$. We now argue that $P(\pi') > P(\pi)$.
%Let $i_{\tau} = \{x,y\}$ and $j_{\tau} = \{a,b\}$.
Note that, $P(\pi') - P(\pi) = \sum_{p \in i_{\tau} \cup j_{\tau} \cup
  \{i,j\}} (\sigma_p(\pi') - \sigma_p(\pi))$. 

To simplify our proof,
we assume edge weights to be non-negative. The cases where some edge
weights are negative follow similarly.  We have the following cases:
	\begin{itemize}
		\item $\nexists p \in i_{\tau}, q \in j_{\tau}$ such that $p = q$: In this case,
		we can treat the swap as independent moves. Thus, it is sufficient to analyze 
		the changes due to a single player. In the following cases, we show that
		$\sum_{p \in i_{\tau} \cup \{i\}} (\sigma_p(\pi')
		- \sigma_p(\pi)) > 0$.  %There are three cases:
		\begin{itemize}
			\item \textit{Case 1:} $i_l \notin N(i,\pi)$ and $i_r \notin N(i,\pi)$. By the definition of the 
			blocking pair resolution, either $i_l \in N(i,\pi')$ or $i_r \in N(i,\pi')$. (Note that, if $i$'s degree is 1, $i_l \in N(i, \pi')$ must hold.)
			\item \textit{Case 2:} Without loss of generality, $i_l \in N(i, \pi)$, $i_r \notin N(i,\pi)$
			and $i_l,i_r \in N(i, \pi')$. Note that this case can only be applicable for players with degree 2. 
			\noindent It can be verified that in both these cases, $\sum_{p \in i_{\tau} \cup \{i\}} (\sigma_p(\pi') - \sigma_p(\pi)) > 0$
			\item \textit{Case 3:} Without loss of generality, $i_l \in N(i,\pi)$, $i_r \notin N(i,\pi)$ and $i_l \notin N(i,\pi')$, $i_r \in N(i,\pi')$, implying 
			$u_i(\pi) = c_{i,2} = w_{i_l,i}$ and $u_i(\pi') = c_{i,1} = w_{i_r,i}$. Note that, $\forall p \in \{i,i_r\}$ $f_p(u_p(\pi')) - f_p(u_p(\pi)) \geq 1$
			while $f_{i_l}(u_{i_l}(\pi')) - f_{i_l}(u_{i_l}(\pi)) \in \{-1,-2\}$. Note that this case can only be applicable for players with degree 2.
			\begin{itemize}
				\item If $f_{i_l}(u_{i_l}(\pi')) - f_{i_l}(u_{i_l}(\pi)) = -1$, then $\sum_{p \in i_{\tau} \cup \{i\}} (\sigma_p(\pi')
				- \sigma_p(\pi))$ $\geq (2n - g(i_r)) + (2n - g(i)) - (2n - g(i_l)) > 0$ as $1 \leq g(k) \leq n$ $\forall k \in N$.
				\item If $f_{i_l}(u_{i_l}(\pi')) - f_{i_l}(u_{i_l}(\pi)) = -2$, it implies $c_{{i_l},1} = w_{i,{i_l}}$. By the definition of the blocking pair
				resolution, $c_{i,1} > c_{{i_l},1}$. Since $w_{i,{i_r}} > c_{{i_l},1}$, we have $c_{{i_r},1} > c_{{i_l},1}$. By definition of $g$, $g({i_l}) > g({i_r})$
				and $g({i_l}) > g(i)$. Therefore, $\sum_{p \in i_{\tau} \cup \{i\}} (\sigma_p(\pi')
				- \sigma_p(\pi))$ $\geq (2n - g({i_r})) + (2n - g(i)) - 2(2n - g({i_l})) = 2g({i_l}) - (g({i_r})+g(i)) > 0$.
				%    $\sum_{p \in i_{\tau} \cup \{i\}} (\sigma_p(\pi')
				% - \sigma_p(\pi))$ $\geq (2n - g(y)) + (2n - g(i)) - 2(2n - g(x)) = 2g(x) - (g(y)+g(i)) > 0$. (Since $f_x(u_x(\pi')) - f_x(u_x(\pi)) = -2$,
				%   it implies $c_{x,1} = w_{i,x}$. By  Therefore, $c_{i,1} > c_{x,1}$. Since $w_{i,y} > c_{x,1}$, $c_{y,1} > c_{x,1}$. By definition of $g$, $g(x) > g(y)$
				%   and $g(x) > g(i)$.)
			\end{itemize}
		\end{itemize} 
		\item $\exists p \in i_{\tau}, q \in j_{\tau}$ such that $p = q$. Without loss of generality, let $i_l = j_l$ and $i_r \neq j_r$ (if $i_r, j_r$ exist):
		\begin{itemize}
			\item \textit{Case 4:} $i_r \notin N(i,\pi)$ and $i_l \notin N(i,\pi)$. If $i_r \in N(i,\pi')$ and $i_l \notin
			N(i,\pi')$, this is the same as Case 1. If $i_l \in N(i, \pi')$, 
			$i_l \in N(j, \pi)$ and $i_l \notin N(j, \pi')$. By the definition of the blocking pair resolution, $j_r \in N(j, \pi')$ and $j_r \notin N(j, \pi)$.
			Therefore, $\forall p \in \{i,j,j_r\}$ $f_p(u_p(\pi')) - f_p(u_p(\pi)) \geq 1$ while $f_{i_l}(u_{i_l}(\pi')) - f_{i_l}(u_{i_l}(\pi)) \geq -1$.
			Thus, $P(\pi') - P(\pi) > 0$.
			\item \textit{Case 5:} $i_r \notin N(i, \pi)$ and $i_l \in N(i,\pi)$, with $i_r \in N(i, \pi')$ and $i_l \in N(i, \pi')$:
			Therefore, $i_l \in N(j, \pi)$. By the definition of the blocking pair resolution, $j_r \in N(j, \pi')$ and $j_r \notin N(j, \pi)$. Therefore, 
			$\forall p \in \{i,j,i_r,j_r\}$ $f_p(u_p(\pi')) - f_p(u_p(\pi)) \geq 1$ while $f_{i_l}(u_{i_l}(\pi')) - f_{i_l}(u_{i_l}(\pi)) \geq 0$.
			Thus, $P(\pi') - P(\pi) > 0$.
			%\item \textit{Case 6:} $x \notin N(i, \pi)$ and $y \in N(i,\pi)$, with $x \in N(i, \pi')$ and $y \notin N(i, \pi')$:
		\end{itemize}
		\item $i_{\tau} = j_{\tau}$ with $|i_{\tau}| = 2$. In this case, it is easy to verify that a blocking pair is only possible if both $i,j$ are connected to one item initially. Let $\pi \to_{i,j} \pi'$. 
		Therefore, for all $p \in \{i,j,i_l,i_r\} (u_p(\pi') - u_p(\pi)) > 0$ implying 
		$f_p(u_p(\pi')) - f_p(u_p(\pi)) > 0$ $\forall p \in \{i,j,i_l,i_r\}$. Thus, $P(\pi') - P(\pi) > 0$. 
%	\begin{itemize}
%		\item \textit{Case 6:} $i_l \in N_(i, \pi)$ and $i_r \in N(j, \pi)$. By the definition of blocking pair, $w_{i_r, i} > w_{i_l, i}$ and $w_{i_l, j} > w_{i_r, j}$. Therefore, $\forall p \in \{i,j,i_l,j_l\}$ $f_p(u_p(\pi')) - f_p(u_p(\pi)) \geq 1$. Thus, $P(\pi') - P(\pi) > 0$.
%		   
%	\end{itemize}
	\end{itemize}
\end{proof}

\thmSymBip*
\begin{proof}
For a symmetric \name{}, the potential function is given by $\phi(\pi)
= \sum_{i \in N} (v_i(\pi) + u_i(\pi))$
(Theorem \ref{thm:sym-potential}). Let $N_U = \{i|\pi(i) \in U\}$ and
$N_V = \{i|\pi(i) \in V\}$. Note that, for a bipartite graph, this can
be written as \\ $\phi(\pi) = \sum_{i \in N_U} (v_i(\pi) + u_i(\pi))
+ \sum_{i \in N_V} (v_i(\pi) + u_i(\pi))$ \\ $= \sum_{i \in N_U}
(2v_i(\pi) + r_i(\pi)) + \sum_{i \in N_V} (2v_i(\pi) + r_i(\pi))$ \\
$= \sum_{i \in N_U} (2v_i(\pi) + 2r_i(\pi)) + \sum_{i \in N_V}
2v_i(\pi)$ \\ $= \sum_{i \in N_U} 2u_i(\pi) + \sum_{i \in N_V}
2v_i(\pi)$ \\ $= \sum_{i \in N_U} 2v_i(\pi) + \sum_{i \in N_V}
2u_i(\pi)$. \\
%
%$\phi(\pi) = \sum_{\{i \mid \pi(i) \in U\}} (v_i(\pi) + u_i(\pi)) + \sum_{\{i \mid \pi(i) \in V\}} (v_i(\pi) + u_i(\pi))$ \\ 
%$= \sum_{\{i \mid \pi(i) \in U\}} (2v_i(\pi) + r_i(\pi)) + \sum_{\{i \mid \pi(i) \in V\}} (2v_i(\pi) + r_i(\pi))$ \\
%$= \sum_{\{i \mid \pi(i) \in U\}} (2v_i(\pi) + 2r_i(\pi)) + \sum_{\{i \mid \pi(i) \in V\}} 2v_i(\pi)$ \\
%$= \sum_{\{i \mid \pi(i) \in U\}} 2u_i(\pi) + \sum_{\{i \mid \pi(i) \in V\}} 2v_i(\pi)$ \\
%$= \sum_{\{i \mid \pi(i) \in U\}} 2v_i(\pi) + \sum_{\{i \mid \pi(i) \in V\}} 2u_i(\pi)$.
Without loss of generality, let $|U| \leq |V|$. For each assignment
$\pi$, we can find the maximum value of $\sum_{i \in N_V} 2v_i(\pi)$
by finding the maximum weight matching of the bipartite graph $Q = (V,
N_V, V \times N_V)$ with the weight of an edge $(i, a)$ being
$v_i(a)$ (in $O(n^4)$).  
Note that there are $\frac{n!}{(n-|U|)!}$ ways of assigning
the $|U|$ items to players in $N$.  For each of these, $\sum_{i \in
  N_U} 2u_i(\pi)$ is fixed since $H$ is a complete bipartite graph,
and the maximum value of $\sum_{i \in N_V} 2v_i(\pi)$ can be computed
by reducing it to a maximum weight matching problem as described above. Thus, an allocation with the optimal potential value can be computed in
$O(n^{\text{min}(|U|,|V|)+4})$. 
%Since the assignment with the optimal
%potential value must be a 2-stable allocation, a 2-stable allocation
%can be computed in $O(n^{\text{min}(|U|,|V|)+4})$.
\end{proof}

%% Stability in asymmetric neighbourhood

\thmCheckTwoStableNP*
\begin{proof}
Given an allocation $\pi$, deciding whether $\pi$ is a 2-stable
allocation can be done in polynomial time. It suffices to check if
there is a blocking pair in $\pi$. Thus the above problem is in NP. To
show hardness, we give a reduction from 3-SAT.
Let the 3-SAT instance have $q$ variables ($\{a_1,a_2,\ldots,a_q\}$) and $m$ clauses ($\{c_1,c_2,\ldots, c_m\}$). We 
create an instance of the \name{} with $|N| = 4m+2q$, with the players 
being differentiated into 6 types with
$N = C_1 \cup C_2 \cup C_3 \cup C_4 \cup S_1 \cup S_2$, where $\forall i \in \{1,2,3,4\}$ $|C_i| = m$
and $\forall j \in \{1,2\}$ $|S_j| = q$.
The key idea is to set up item valuations and the player neighbourhood 
structure such that a stable allocation can only exist if each player is 
assigned an item from the item set corresponding to its type(s) \textit{and} 
the 3-SAT instance is satisfiable. We set up the following constants to aid
our explanation: $b = (2q + 3m + 20)d = (2q + 3m + 20)^{2}e = (2q + 3m + 20)^{3}f$ with $f = 1$.
There are 4 players corresponding to 
each clause and 2 players corresponding to each variable. For each clause $c_i$, 
the players corresponding to it ($c_{i,1}$,$c_{i,2}$, $c_{i,3}$ and $c_{i,4}$)
are connected as shown in 
Figure \ref{fig:clause}, with $w_{c_{i,2}, c_{i,1}} = w_{c_{i,4}, c_{i,3}} = d - e$,
$w_{c_{i,1}, c_{i,2}} = w_{c_{i,3}, c_{i,4}} = d$, 
$w_{c_{i,4}, c_{i,1}} = w_{c_{i,2}, c_{i,3}} = d$, 
$w_{c_{i,1}, c_{i,4}} = w_{c_{i,3}, c_{i,2}} = d - e$ and 
$w_{c_{i,3}, c_{i,1}} = w_{c_{i,1}, c_{i,3}} = w_{c_{i,2}, c_{i,4}} = w_{c_{i,4}, c_{i,2}} = -2d$. 
For each variable $a_j$, there are two
players $s_{j,1}$ and $s_{j,2}$ corresponding to the positive literal $a_j$
and the negative literal $\neg a_j$ respectively. 
The edge connecting these two players has a large negative weight
($w_{s_{j,1},s_{j,2}} = w_{s_{j,2},s_{j,1}} = -d$). 
For each clause $c_i$ where a positive literal $a_j$ appears $x$ times
and the corresponding negative literal $\neg a_j$ appears $y$ times,
$w_{s_{j,1}, c_{i,1}} = (x - y)e$ and $w_{s_{j,2}, c_{i,1}} = (y - x)e$.
For example, suppose $c_i = a_{t} \vee \neg a_{u} \vee a_{v}$. ($t \neq u \neq v$) Then, 
$w_{s_{t,1}, c_{i,1}} = w_{s_{u,2}, c_{i,1}} = w_{s_{v,1}, c_{i,1}} = e$ and 
$w_{s_{t,2}, c_{i,1}} = w_{s_{u,1}, c_{i,1}} = w_{s_{v,2}, c_{i,1}} = -e$.

The item graph is $A = A_{C_1} \cup A_{C_2} \cup A_{C_3} \cup A_{C_4} \cup A_{S_1} \cup A_{S_2}$, 
with $|A_{C_i}| = m$ and $|A_{S_j}| = q$, with each of these six subsets being cliques.
Additionally, there is complete connection between $A_{S_2}$ and $A_{C_1}$, 
$A_{C_1}$ and $A_{C_2}$ and $A_{C_3}$ and $A_{C_4}$.  
The table above contains $v_p(t)$, for each player $p$ and item $t$.
%the valuation of item $t$ for player $p$.

	%\begin{tabular}{|l|l|l|l|l|l|}
	\begin{table}
        \centering
		%\label{FTable} 
		\resizebox{0.45\textwidth}{!}{%
			%\begin{tabular}{|c|@{}c@{}|@{}c@{}|@{}c@{}|@{}c@{}|@{}c@{}|}
			\begin{tabular}{|l|l|l|l|l|l|}
				%\begin{tabular}{|c@{}|c@{}|c@{}|c@{}|c@{}|c@{}|}
				\hline
				& $p \in S_1 \cup S_2$&$p \in C_1$&$p \in C_2$&$p \in C_3$&$p \in C_4$ \\
				\hline
				$t \in A_{S_1}$ & $b$ & 0 & 0 & 0 & 0 \\
				\hline
				$t \in A_{S_2}$ & $b$ & 0 & 0 & 0 & 0 \\
				\hline
				$t \in A_{C_1}$ & 0 & $b+4e-f$ & $b-6d$ & $b$ & $b-6d$ \\
				\hline
				$t \in A_{C_2}$ & 0 & $b-6d$ & $b$ & $b-6d$ & $b - f$ \\
				\hline
				$t \in A_{C_3}$ & 0 & $b-f$ & $b-6d$ & $b$ & $b-6d$ \\
				\hline
				$t \in A_{C_4}$ & 0 & $b-6d$ & $b-f$ & $b-6d$ & $b$ \\
				% cell4 & cell5 & cell6 \\ 
				% cell7 & cell8 & cell9 \\ 
				\hline
				
			\end{tabular}
		}
	\end{table} 

Let us assume that the 3-SAT instance has some satisfying assignment. We can construct a 2-stable allocation
using this satisfying assignment. Let $T$ denote the set of variables which are set to true in this satisfying 
assignment, and $F$ denote the set of variables set to false in this assignment. It is easy to verify that the
following allocation $\pi$ is 2-stable: $\forall i \in \{1,2,3,4\}$ $\forall j \in C_i$ $\pi(j) \in A_{C_i}$, $\forall a_j \in T$ 
$\pi(s_{j,1}) \in A_{S_2}$ and $\pi(s_{j,2}) \in A_{S_1}$ and $\forall a_j \in F$ 
$\pi(s_{j,2}) \in A_{S_2}$ and $\pi(s_{j,1}) \in A_{S_1}$. Note that, all players in $S_1 \cup S_2 \cup C_2 \cup C_4$ have the maximum utility 
attainable, and hence, cannot be a part of any blocking pair. Thus, the only possible blocking pair is between players $i,j \in C_1 \cup C_3$.
Since any two players $x,y \in C_i$ are connected to the same set of players, and their valuations for all items in $A_{C_i}$ are the same, 
they cannot form a blocking pair. Thus, the only possible blocking pair is between players $x, y$ where, without loss of generality, $x \in C_1$
and $y \in C_3$. Since the assignment is satisfying, $v_{x}(\pi(x)) \geq b+d+e-f > b+d \geq v_{x}(\pi(y))$. Hence, this is a 2-stable allocation.

To complete the proof, we simply need to show that every 2-stable allocation will generate a satisfying assignment for the 3-SAT instance. We first list necessary conditions for a 2-stable allocation.
\begin{claim}
	\label{cl1}
	For every 2-stable allocation $\pi$, it is necessary that: 
	\begin{itemize}
		\item  For all $p \in S_1 \cup S_2$, $\pi(p) \in A_{S_1} \cup A_{S_2}$ 
		\item For all $i \in \{1,2,3,4\}$, for all $p \in C_i$, $\pi(p) \in A_{C_i}$.
		\item There does not exist $i$ s.t. $\pi(s_{i,1}) \in S_k$ and $\pi(s_{i,2}) \in S_k$ for some $k \in \{1,2\}$. 
	\end{itemize}
\end{claim}
The proof of the claim follows from the lemmas given below:
\begin{lemma}
	\label{StableL1}
	For any 2-stable allocation $\pi$, $\forall x \in S_1 \cup S_2$, $\pi(x) \in A_{S_1} \cup A_{S_2}$ and, 
	$\forall i \in \{1,2,3,4\}$ $\forall y \in C_i$ $\pi(y) \in A_{C_1} \cup A_{C_2} \cup A_{C_3} \cup A_{C_4}$.
\end{lemma}
\begin{lemma}
	\label{StableL2}
	For any 2-stable allocation $\pi$, for each player $p \in C_1 \cup C_3$, $\pi(p) \in A_{C_1} \cup A_{C_3}$ and for each player $q \in C_2 \cup C_4$, $\pi(q) \in A_{C_2} \cup A_{C_4}$.
\end{lemma}
\begin{lemma}
	\label{StableL3}
	For any 2-stable allocation $\pi$, for each player $p \in C_1$, $\pi(p) \in A_{C_1}$ and,
	for each player $q \in C_3$, $\pi(q) \in A_{C_3}$.
\end{lemma}
\begin{lemma}
	\label{StableL4}
	For any 2-stable allocation $\pi$, for each player $p \in C_2$, $\pi(p) \in A_{C_2}$ and, for each player $q \in C_4$, $\pi(q) \in A_{C_4}$.
\end{lemma}
\begin{lemma}
	\label{StableL5}
	For any 2-stable allocation $\pi$, for each pair of players $s_{i,1}$ and $s_{i,2}$, 
	if $\pi(s_{i,1}) \in A_{S_k}$, $\pi(s_{i,2}) \in A_{S_{3-k}}$ for $k \in \{1,2\}$.
\end{lemma}
We can generate an assignment for the 3-SAT instance using a 2-stable allocation as follows:
Let $X = \{i \in S_1 \cup S_2 \mid \pi(i) \in A_{S_2}\}$. Since the allocation is 2-stable, assuming Claim \ref{cl1}, 
for each variable $a_j$ in the 3-SAT instance, exactly one of $s_{j,1}$ and $s_{j,2}$ must be in $X$. 
If $s_{j,1} \in X$, we set $a_j$ to true, otherwise, we set it to false. We will show that this constitutes 
a satisfying assignment for the 3-SAT instance.
Let $X_i = $ Number of satisfying literals for clause 
$c_i$ in the generated assignment and $Y_i = $ Number of 
unsatisfying literals for clause $c_i$ in the generated assignment.
%Setting the item valuations and the neighbourhood payoffs appropriately, 
%this instance has a stable allocation iff the 3-SAT instance is satisfiable.
Suppose there exists some 2-stable allocation $\pi$. Let us assume, 
to the contrary, that the generated assignment does not satisfy the 3-SAT instance.   
Therefore, there must exist some clause $c_i$ which is not satisfied. 
However, note that $u_{c_{i,1}}(\pi) = (b+d+3e-f) + e(X_i - Y_i)$ and $u_{c_{i,3}}(\pi) = b+d-e$,
while $u_{c_{i,1}}(\pi') = b+d$ and $u_{c_{i,3}}(\pi') = b+d-f$ where $\pi'(c_{i,1}) = \pi(c_{i,3})$, 
$\pi'(c_{i,3}) = \pi(c_{i,1})$ and $\forall x \in ((S_1 \cup S_2 \cup C_1 \cup C_2 \cup C_3 \cup C_4) - \{c_{i,1},c_{i,3}\})$ $\pi'(x) = \pi(x)$.
But, this implies that $c_{i,1}$ and $c_{i,3}$ form a blocking pair with $\pi \to_{c_{i,1},c_{i,3}} \pi'$. Thus, $\pi$ is not 2-stable, 
leading to the contradiction.
\end{proof}

The proofs of Lemmas \ref{StableL1} and \ref{StableL2} are
straightforward. We give the proofs of
Lemmas \ref{StableL3}, \ref{StableL4} and \ref{StableL5} below.

\begin{proof}[Proof of Lemma \ref{StableL3}]
Suppose there exists some player $c_{i,1} \in C_1$ such that $\pi(c_{i,1}) \in A_{C_3}$. By lemma  \ref{StableL2}, there must exist some player $c_{j,3} \in C_3$ such that $\pi(c_{j,3}) \in A_{C_1}$. There are two cases:
\begin{itemize}
	\item $\exists i, j$ such that $i = j$: There are 3 cases now:
	\begin{itemize}
		\item $\pi(c_{j,2}) \in A_{C_k}$ and $\pi(c_{j,4}) \in A_{C_k}$ for some $k \in \{2,4\}$: 
		Since $\pi$ is 2-stable, there must be some $l$ such that 
		$\pi(c_{l,2}) \in A_{C_{6-k}}$ and $\pi(c_{l,4}) \in A_{C_{6-k}}$. 
		But, in this case, $c_{j,2}$ and $c_{l,4}$ will form a blocking pair with 
		$\pi \to_{c_{j,2},c_{l,4}} \pi'$ as $u_{c_{j,2}}(\pi) \leq b - d < b - f \leq u_{c_{j,2}}(\pi')$.
		Similarly, $u_{c_{l,4}}(\pi) \leq b - e < b - f \leq u_{c_{l,4}}(\pi')$. 
		Hence, this is not a 2-stable allocation.
		\item $\pi(c_{j,2}) \in A_{C_2}$ and $\pi(c_{j,4}) \in A_{C_4}$: In this case, $c_{j,2}$ and $c_{j,4}$ form a blocking pair, with
		$\pi \to_{c_{j,2},c_{j,4}} \pi'$ as $u_{c_{j,2}}(\pi) \leq b+d-e < b+d-f \leq u_{c_{j,2}}(\pi')$ and
		$u_{c_{j,4}}(\pi) \leq b+d-e < b+d-f \leq u_{c_{j,4}}(\pi')$. Thus, this is not a 2-stable allocation.
		\item $\pi(c_{j,2}) \in A_{C_4}$ and $\pi(c_{j,4}) \in A_{C_2}$: In this case, $c_{j,1}$ and $c_{j,3}$ form a blocking pair, 
		with $\pi \to_{c_{j,1},c_{j,3}} \pi'$ as $u_{c_{j,3}}(\pi)) \leq b+d-(e+f) < b+d \leq u_{c_{j,3}}(\pi')$ and
		$u_{c_{j,1}}(\pi) \leq b+d-e < b+d+e-f \leq u_{c_{j,1}}(\pi')$. Thus, this is not a 2-stable allocation.
	\end{itemize}
	\item $\nexists i,j$ such that $i \neq j$:
	% (Note: We assume $i \neq j$ only if $\pi^{-1}(c_{3,i}) \notin A_{C_1}$ and $\pi^{-1}(c_{1,j}) \notin A_{C_3}$, since, otherwise, we would be able to find some $k$ such that $\pi^{-1}(c_{1,k} \in A_{C_3}$ and $\pi^{-1}{c_{3,k}} \in A_{C_1}$))
	In this case, $c_{i,3}$ and $c_{j,1}$ form a blocking pair with $\pi \to_{c_{i,3},c_{j,1}} \pi'$ as 
	$u_{c_{i,3}}(\pi) \leq b-(e+f) < b-f \leq u_{c_{i,3}}(\pi')$ and $u_{c_{j,1}}(\pi) \leq b-(e+f) < b+e-f \leq u_{c_{j,1}}(\pi')$. Thus, this is not a 2-stable allocation.
\end{itemize}
\end{proof}

\begin{proof}[Proof of Lemma \ref{StableL4}]
Suppose there exists some player $c_{i,2} \in C_2$ such that $\pi(c_{i,2}) \in A_{C_4}$.  By lemma \ref{StableL2}, there must exist some $c_{j,4} \in C_4$ such that $\pi(c_{j,4}) \in A_{C_2}$. There are two cases:
\begin{itemize}
	\item $\exists i,j$ such that $i = j$: In this case, $c_{j,2}$ and $c_{j,4}$ form a blocking pair, 
	with $\pi \to_{c_{j,2},c_{j,4}} \pi'$ as 
	$u_{c_{j,2}}(\pi) \leq b+d-(e+f) < b+d \leq u_{c_{j,2}}(\pi')$ and
	$u_{c_{j,4}}(\pi) \leq b+d-(e+f) < b+d \leq u_{c_{j,4}}(\pi')$. 
	Hence, this allocation is not 2-stable.
	\item $\nexists i,j$ such that $i \neq j$: 
	% (Note: We assume $i \neq j$ only if $\pi^{-1}(c_{4,i}) \notin A_{C_2}$ and $\pi^{-1}(c_{2,j}) \notin A_{C_4}$, since, otherwise, we would be able to find some $k$ such that $\pi^{-1}(c_{1,k} \in A_{C_3}$ and $\pi^{-1}{c_{3,k}} \in A_{C_1}$))
	In this case, $c_{i,2}$ and $c_{j,4}$ form a blocking pair with $\pi \to_{c_{i,2},c_{j,4}} \pi'$ as 
	$u_{c_{i,2}}(\pi) \leq b-(d+e+f) < b+d \leq u_{c_{i,2}}(\pi')$ and 
	$u_{c_{j,4}}(\pi) \leq b-(d+e+f) < b+d \leq u_{c_{j,4}}(\pi')$. 
	Thus, this is not a 2-stable allocation. 
	% In this case, $c_{2,i}$ and $c_{4,j}$ form a blocking pair. (The payoff of $c_{2,i}$ in the current position is $b - (d + e + f)$ while the payoff in the new position is $b + d$. Similarly, the payoff of $c_{4,j}$ in the current position is $b - (d + e + f)$ while the payoff in the new position is $b + d$.) Hence, this allocation is not 2-stable. 
\end{itemize}
\end{proof}

\begin{proof}[Proof of Lemma \ref{StableL5}]
	Suppose, $\exists s_{i,1}, s_{i,2}$ such that $\pi(s_{i,1}) \in A_{S_k}$ 
	and $\pi(s_{i,2}) \in A_{S_k}$ for some $k \in \{1,2\}$. Since $\pi$ is 
	2-stable, by lemma \ref{StableL1}, $\exists s_{j,1}, s_{j,2}$ such that
	$\pi(s_{j,1}) \in A_{S_{3-k}}$ and $\pi(s_{j,2}) \in A_{S_{3-k}}$. But, $s_{j,1}$ and $s_{i,1}$
	form a blocking pair with $\pi \to_{s_{j,1}, s_{i,1}} \pi'$, as $\forall x \in \{1,2\}$ 
	$u_{s_{j,x}}(\pi) \leq b - d < b \leq u_{s_{j,x}}(\pi')$.
\end{proof}

%% Postitive results

\thmTwoStableDAG*
\begin{proof}
% To compute the allocation, we first do a topological sort of the
% player graph. We then do an assignment in this topological order by
% allowing each player their most preferred item from those items which
% have not been allocated so far. If a player has multiple equally
% preferred items, they are allocated one of the items randomly. Let
% $\pi$ be such an allocation.  Let us now assume, to the contrary, that
% $\pi$ is not 2-stable. Therefore, there must exist two players (say
% $p$ and $q$) such that $u_p(\pi^{-1}(q)) > u_p(\pi^{-1}(p))$ and
% $u_q(\pi^{-1}(p)) > u_q(\pi^{-1}(q))$. Without loss of generality, let
% us assume that $p$ is allocated an item first. Note that, at the point
% when $p$ is allocated an item, all of the players on whom $p$'s
% neighbourhood valuation is dependent have been assigned an item. Thus,
% on choosing an item, $p$'s valuation is fixed. But, at this point,
% both $\pi^{-1}(p)$ and $\pi^{-1}(q)$ are available. Hence, by our
% algorithm, $p$ must pick $\pi^{-1}(q)$. Hence, there is a
% contradiction.
%%
To compute the allocation, we first do a topological sort of the
player graph and assign items in this topological ordering.  For each
player, we assign the item from amongst the unassigned items which
maximizes the player's net utility. If there are multiple such items,
we choose an arbitrary item.  Let $\pi$ be a resulting
allocation. Suppose $\pi$ is not core stable. Then, there exists a
blocking coalition $X$ with $\pi'$ and $\mu:X \to X$ such that
$\forall x \in X$ $\pi'(x) \neq \pi(x), \pi'(x) = \pi(\mu(x))$ and
$u_x(\pi'(x)) > u_x(\pi(x))$. Let $y$ be the first player from $X$ in
the topological ordering. Therefore, $u_y(\pi'(y)) >
u_y(\pi(y))$. Note that, since $\pgraph$ is a DAG, the utility of $y$
depends only on the items allocated to players before $y$ in the
topological order. But, since $y$ is the first player in $X$ to be
allocated an item, both $\pi'(y)$ and $\pi(y)$ are available at the
time $y$ is allocated an item. Since $y$ chooses the item which
maximizes its utility, $u_y(\pi'(y)) > u_y(\pi(y))$ is a
contradiction.
\end{proof}

\thmCycleOneConnectedNode*
\begin{proof}
Let the players in the cycle, in order, be numbered from 1 to $n$, and
$w_{j-1,j} = c_j$ $\forall j \in \{2,3,\ldots,n\}$ and $w_{n,1} =
c_1$. Without loss of generality, let the first player be one with a
non-negative neighbourhood valuation. That is, $w_{n,1} = c_1$, $c_1 >
0$. We assign the first player an item connected to all of the
remaining items. We assign the remaining items in the order
$\{2,3,\ldots,n\}$. For each player, we assign the item from amongst
the unassigned items which maximizes the player's net utility. If
there are multiple such items, we choose an arbitrary item. Let $\pi$
be a resulting allocation. Suppose, $\pi$ is not 2-stable. Therefore,
$\exists x,y$ such that $\pi \to_{x,y} \pi'$. Without loss of
generality, let $x$ be the player assigned an item first.  There are
two cases:
\begin{itemize}
\item $x = 1$: Since the item valuations are uniform, the only difference in utility is because of the neighbourhood structure. Since $a = \pi(1)$ is connected to each of the items and $w_{n,1} > 0$, $u_1(\pi'(1)) \leq u_1(\pi(1))$. Hence, there is a contradiction. 
\item $x \neq 1$: Since player $x-1$ has already been allocated an item,
the utility of $x$ is fixed on choosing an item. Since $(x,y)$ is a blocking pair, $u_x(\pi'(x)) > u_x(\pi(x))$. 
But, when $x$ is allocated an item, both $\pi(x)$ and $\pi'(x)$ are available. Since $x$ chooses the item which 
maximizes its utility, $u_x(\pi'(x)) > u_x(\pi(x))$ is a contradiction.
% At the time of allocation of the item for $p$, note that both $\pi^{-1}(p)$ and $\pi^{-1}(q)$ are available. Also, at this point, player $p-1$ has already been allocated an item. Thus, $p$'s valuation is fixed on choosing an item. Also, by our algorithm, $p$ chooses the item with maximum net valuation. Hence, $u_p(\pi^{-1}(q)) \leq u_p(\pi^{-1}(p))$. Hence, there is a contradiction.
\end{itemize}
% Therefore, there must exist players $p$ and $q$, such that $u_p(\pi^{-1}(q)) > u_p(\pi^{-1}(p))$ and $u_q(\pi^{-1}(p)) > u_q(\pi^{-1}(q))$. Without loss of generality, let $p$ be the player assigned an item first. There are two cases:
% \begin{itemize}
% \item $p = 1$: Note that, the valuations are uniform. Therefore, the only valuation difference is because of the neighbourhood valuation. Since $\pi^{-1}(1)$ is connected to each of the items. Hence, $u_p(\pi^{-1}(q)) <= u_p(\pi^{-1}(p))$. Hence, there is a contradiction. 
% \item $p \neq 1$: At the time of allocation of the item for $p$, note that both $\pi^{-1}(p)$ and $\pi^{-1}(q)$ are available. Also, at this point, player $p-1$ has already been allocated an item. Thus, $p$'s valuation is fixed on choosing an item. Also, by our algorithm, $p$ chooses the item with maximum net valuation. Hence, $u_p(\pi^{-1}(q)) \leq u_p(\pi^{-1}(p))$. Hence, there is a contradiction.
% \end{itemize}
\end{proof}

\thmCycleOneDegreeItemCoreStable*
\begin{proof}
Let the players in the cycle, in order, be numbered from 1 to $n$. We
assign the first player an item with degree at most 1. We assign the remaining items in the order
$\{2,3,\ldots,n\}$. For each player, we assign the item from amongst the
unassigned items which maximizes the player's utility. If there
are multiple such items, we choose an arbitrary item. Let $\pi$ be a resulting 
allocation. Suppose $\pi$ is not core stable. Then, there exists a blocking coalition $X$ 
with $\pi'$ and $\mu:X \to X$ such that $\forall x \in X$ $\pi'(x) \neq \pi(x), \pi'(x) = \pi(\mu(x))$
and $u_x(\pi'(x)) > u_x(\pi(x))$. Let $y$ be the first player 
from $X$ in the ordering. Therefore, $u_y(\pi'(y)) > u_y(\pi(y))$. There are two cases:
\begin{itemize}
\item $y = 1$: Since the item valuations are uniform, the only difference in utility is because of the neighbourhood structure. 
Since $a = \pi(1)$ has degree one, it must be connected to a single item. Since player two is the next player to be allocated an item, 
and all connection weights are positive, the item connected to $\pi(1)$ (say $b$) must be allocated to player two. Additionally, 
since the graph is connected, $\pi(2)(=b)$ must be connected to at least one node other than $\pi(1)$ (say $c$). Since player three is the next player
to be allocated an item, it must be allocated an item $c$. Since players two and three have their maximum possible utility under $\pi$, and $\pi(1)$
is connected only to $\pi(2)$, there can be no player $x$ such that $\mu(x) = 1$, and $u_x(\pi(1)) > u_x(\pi(x))$. Hence, there is a contradiction.
\item $y \neq 1$: Since player $y-1$ has already been allocated an item,
the utility of $y$ is fixed on choosing an item. 
But, since $y$ is the first player in $X$ to be
allocated an item, both $\pi'(y)$ and $\pi(y)$ are available at the time
$y$ is allocated an item. Since $y$ chooses the item which maximizes its utility,
$u_y(\pi'(y)) > u_y(\pi(y))$ is a contradiction. 
\end{itemize}
\end{proof}

%% Envy-freeness

\CycleDirectedEF*

\begin{proof}
It is easy to see that the problem is in NP. We show NP-hardness by
giving a reduction from the Hamiltonian cycle problem. Consider an
instance of the Hamiltonian cycle problem given by the connected graph
$Q= (T, E)$.  We construct an instance of the \name{} as follows: the item
graph $\igraph=(A,\lambda)$ is the graph $Q$. The player graph is a
directed \cycle{} $\pgraph=(N,\tau,w)$ where $|N| = |T|$ and for
$i \in N$, $i_{\tau} = \{i-1\}$ where $n+1 = 1$ and $1-1 = n$. The
player graph is unweighted, i.e.\ for all $(i,j) \in \tau$,
$w_{i,j}=1$.  We set $v_i(a)=0$ for all $i \in N$ and $a \in A$.
%% Since the player graph is unweighted, $\forall
%% i \in \{2,3,4,\ldots,n\}$ $w_{i-1,i} = 1$ and $w_{n,1} = 1$.
%%  We take the item valuation functions
%% $V$ to be zero for each player for each item.
% Additionally, $w_{i}(i-1) = 1$
% $\forall i \in \{2,2,3,...n\}$ and $w_1(n) = 1$. (implying that a
% player's neighbourhood utility is through the previous player on the
% cycle)
%% Note that, we can take the undirected graph $Q$ to be
%% connected, as, for disconnected graphs, we can trivially confirm the
%% non-existence of a Hamiltonian cycle and can also confirm the
%% disconnection in polynomial time.
%%
%% \blue{It can be verified that, if there exists a Hamiltonian cycle in $Q$,
%% there will be an envy-free allocation for the players.

If there exists a Hamiltonian cycle in $Q$,
we can construct an envy-free allocation as follows. We allocate items
to the players from $1$ to $n$ in the order given by the Hamiltonian
cycle. We now argue that the existence of an envy-free allocation
implies the existence of a Hamiltonian cycle in $Q$. It suffices to
show that, in an envy-free allocation $\pi$, for all $i \in N$,
$(i-1) \in d_i(\pi)$ (recall that $1-1=n$).
%(with $i-1 = n$ for $i = 1$). 
If each player is connected to their preceding player in the
allocation $\pi$, we can simply start from the item allocated to
player $n$, and construct the cycle by adding the edge to the
allocation of the preceding player. Let us assume, to the contrary,
that there exists some envy-free allocation $\pi$ such that, for some
$i \in N$, $i-1 \not\in d_i(\pi)$. Since the item graph is connected,
there must exist some $j \in N$ such that $i-1 \in
d_j(\pi)$. 
%% Thus, $i$ must envy the item allocated to $j$.
Thus, $u_i(\pi) < u_i(\pi')$ where $\pi'(i) = \pi(j)$, $\pi'(j)
= \pi(i)$ and for all $k \in N -\{i,j\}$, $\pi(k) = \pi'(k)$.  Hence,
$\pi$ is not an envy-free allocation which is a contradiction.
\end{proof}

\end{document}